%% file: main.tex
\documentclass[sigconf,authorversion=true]{acmart}
\input{macros.tex}

\copyrightyear{2020}
\acmYear{2020}
\setcopyright{acmlicensed}\acmConference[KDD '20]{Proceedings of the 26th ACM SIGKDD Conference on Knowledge Discovery and Data Mining}{August 23--27, 2020}{Virtual Event, CA, USA}
\acmBooktitle{Proceedings of the 26th ACM SIGKDD Conference on Knowledge Discovery and Data Mining (KDD '20), August 23--27, 2020, Virtual Event, CA, USA}
\acmPrice{15.00}
\acmDOI{10.1145/3394486.3403383}
\acmISBN{978-1-4503-7998-4/20/08}

\begin{document}
\fancyhead{}

\title[Explainable Classification of Brain Networks via Contrast Subgraphs]{Explainable Classification of Brain Networks via\\Contrast Subgraphs}

\author{Tommaso Lanciano}
\affiliation{%
	\institution{Sapienza University, Rome, Italy}
}
\email{tommaso.lanciano@uniroma1.it}

\author{Francesco Bonchi}
\affiliation{%
	\institution{ISI Foundation, Turin, Italy}
}
\email{francesco.bonchi@isi.it}

\author{Aristides Gionis}
\affiliation{%
\institution{KTH, Sweden}
}
\email{argioni@kth.se}

\begin{abstract}
Mining human-brain networks to discover patterns that can be used to discriminate
between healthy individuals and patients affected by some neurological disorder,
is a fundamental task in neuro\-science.
Learning simple and interpretable models is as important as mere classification accuracy.
In this paper we introduce a novel approach for classifying brain networks
based on extracting \emph{contrast subgraphs},
i.e., a set of vertices whose induced subgraphs are dense in one class of graphs and sparse in the other.
We formally define the problem and present an algorithmic solution for extracting contrast subgraphs.
We then apply our method to a brain-network dataset consisting of children affected by
Autism Spectrum Disorder and children Typically Developed.
Our analysis confirms the interestingness of the discovered patterns,
which match background knowledge in the neuro\-science literature.
Further analysis on other classification tasks confirm the simplicity,
soundness, and high explainability of our proposal,
which also exhibits superior classification accuracy,
to more complex state-of-the-art methods.
\end{abstract}
\maketitle \sloppy

\section{Introduction}
\label{sec:intro}
\input{intro.tex}

\section{Related work}
\label{sec:related}
\input{related.tex}

\section{Problem statement}
\label{sec:problem}
\input{problem.tex}

\section{Complexity and algorithms}
\label{sec:algorithms}
\input{algorithms.tex}

\section{Experiments}
\label{sec:experiments_desc}
\input{experiments.tex}

\input{experiments2.tex}

\section{Conclusions and future work}
\label{sec:conclusions}
\input{conclusions.tex}

\spara{Acknowledgements.}
\input{acknowledgements.tex}


\input{staticbib.tex}
\newpage
\appendix
\section{Reproducibility}
\label{sec:appendix}
\input{appendix.tex}

\end{document}

%% file: macros.tex
\usepackage{graphicx}
\usepackage[show]{chato-notes}
\usepackage{amsmath,amssymb}
\usepackage{epsfig}
\usepackage{algorithm}
\usepackage[noend]{algorithmic}
\usepackage{url}
\usepackage{hyperref}
\usepackage{bbm}
\usepackage[export]{adjustbox}
\usepackage{xspace}
\usepackage{booktabs}
\usepackage{tikz}
\usetikzlibrary{shadows}
\usepackage{xcolor}
\usepackage[framemethod=TikZ]{mdframed}
\usepackage{tkz-berge,float}

\newcommand{\spara}[1]{\smallskip\noindent{\bf #1}}

\newtheorem{proposition}{Proposition}

\newtheorem{myexample}{Example}

\newtheorem{problem}{Problem}

\DeclareMathOperator*{\argmax}{arg\,max}
\newcommand{\bigO}{\mathcal{O}}

\newcommand{\realsnn}{\ensuremath{\mathbb{R}_+}\xspace}

\newcommand{\NPhard}{\ensuremath{\mathbf{NP}}-hard\xspace}
\newcommand{\NPcomplete}{\ensuremath{\mathbf{NP}}-complete\xspace}

\newcommand{\refpr}[1]{Problem~\ref{prb:#1}}
\newcommand{\reffig}[1]{Figure~\ref{fig:#1}}

\newcommand{\reftab}[1]{Table~\ref{tab:#1}}
\newcommand{\refsec}[1]{Section~\ref{sec:#1}}

\newcommand{\diffnet}[2]{\ensuremath{G^{#1-#2}}\xspace}

\renewcommand{\vec}[1]{\mathbf{#1}}

\newcommand{\dataset}{\ensuremath{\mathcal{D}}\xspace}
\newcommand{\A}{\ensuremath{A}\xspace}
\newcommand{\B}{\ensuremath{B}\xspace}
\newcommand{\calA}{\ensuremath{\mathcal{A}}\xspace}
\newcommand{\calB}{\ensuremath{\mathcal{B}}\xspace}
\newcommand{\GA}{\ensuremath{G^\calA}\xspace}
\newcommand{\GB}{\ensuremath{G^\calB}\xspace}
\newcommand{\GTD}{\ensuremath{G^\td}\xspace}
\newcommand{\GASD}{\ensuremath{G^\asd}\xspace}
\newcommand{\eA}{\ensuremath{e^\calA}\xspace}
\newcommand{\eB}{\ensuremath{e^\calB}\xspace}
\newcommand{\wA}{\ensuremath{w^\calA}\xspace}
\newcommand{\wB}{\ensuremath{w^\calB}\xspace}
\newcommand{\wTD}{\ensuremath{w^\td}\xspace}
\newcommand{\wASD}{\ensuremath{w^\asd}\xspace}

\newcommand{\highlight}[1]{\ensuremath{\text{\small\textsf{#1}}}\xspace}
\newcommand{\oqc}{\ensuremath{\text{\small OQC}}\xspace}
\newcommand{\goqc}{\ensuremath{\text{\small GOQC}}\xspace}
\newcommand{\asd}{\ensuremath{\text{\small\textsf{ASD}}}\xspace}
\newcommand{\td}{\ensuremath{\text{\small\textsf{TD}}}\xspace}
\newcommand{\roi}{\ensuremath{\text{ROI}}\xspace}
\newcommand{\mri}{\ensuremath{\text{MRI}}\xspace}

\newcommand{\squishlist}{
 \begin{list}{$\bullet$}
  {  \setlength{\itemsep}{0pt}
     \setlength{\parsep}{3pt}
     \setlength{\topsep}{3pt}
     \setlength{\partopsep}{0pt}
     \setlength{\leftmargin}{2em}
     \setlength{\labelwidth}{1.5em}
     \setlength{\labelsep}{0.5em}
} }
\newcommand{\squishlisttight}{
 \begin{list}{$\bullet$}
  { \setlength{\itemsep}{0pt}
    \setlength{\parsep}{0pt}
    \setlength{\topsep}{0pt}
    \setlength{\partopsep}{0pt}
    \setlength{\leftmargin}{2em}
    \setlength{\labelwidth}{1.5em}
    \setlength{\labelsep}{0.5em}
} }

\newcommand{\squishdesc}{
 \begin{list}{}
  {  \setlength{\itemsep}{0pt}
     \setlength{\parsep}{3pt}
     \setlength{\topsep}{3pt}
     \setlength{\partopsep}{0pt}
     \setlength{\leftmargin}{1em}
     \setlength{\labelwidth}{1.5em}
     \setlength{\labelsep}{0.5em}
} }

\newcommand{\squishend}{
  \end{list}
}

    \renewenvironment{quote}{%
       \list{}{%
         \leftmargin0.2cm   
         \rightmargin\leftmargin
       }
       \item\relax
    }
    {\endlist}

%% file: intro.tex
The development of magnetic resonance imaging ({\mri}) techniques has paved the way to
\emph{connectomics}~\cite{sporns2005human}, i.e., modeling the brain as a network,
allowing to tackle interesting neuro\-science research questions
as graph-analysis problems~\cite{bullmore2009complex}.

A connectome is a map describing neural connections between brain regions of interest ({\roi}s),
either by observing anatomic fiber density (structural connectome), or
by computing pairwise correlations between time series of activity associated to {\roi}s (functional connectome).
The latter approach, known as functional magnetic resonance imaging (f{\mri}),
exploits the link between neural activity and blood flow and oxygenation,
to associate a time series to each {\roi}.
A brain network can then be defined by creating a link between two {\roi}s that exhibit co-activation,
i.e., strong correlation in their time series.
An important problem in this domain is to identify
connection patterns that might be associated to specific cognitive phenotypes
or mental dysfunctions~\cite{du2018brainclassificationsurvey}.
Given f{\mri} scans of patients affected by a mental disorder and scans of healthy individuals,
the goal is to discover patterns in the corresponding connectomes
that explain differences in the brain mechanism of the two groups.
Identifying such patterns might provide important insight of the disorder and hint
strategies to improve the condition of patients.

\begin{figure}[t!]
	\begin{minipage}[c]{.38\linewidth}
		\begin{tabular}{c}
			\hspace{-4mm}\includegraphics[width= \linewidth]{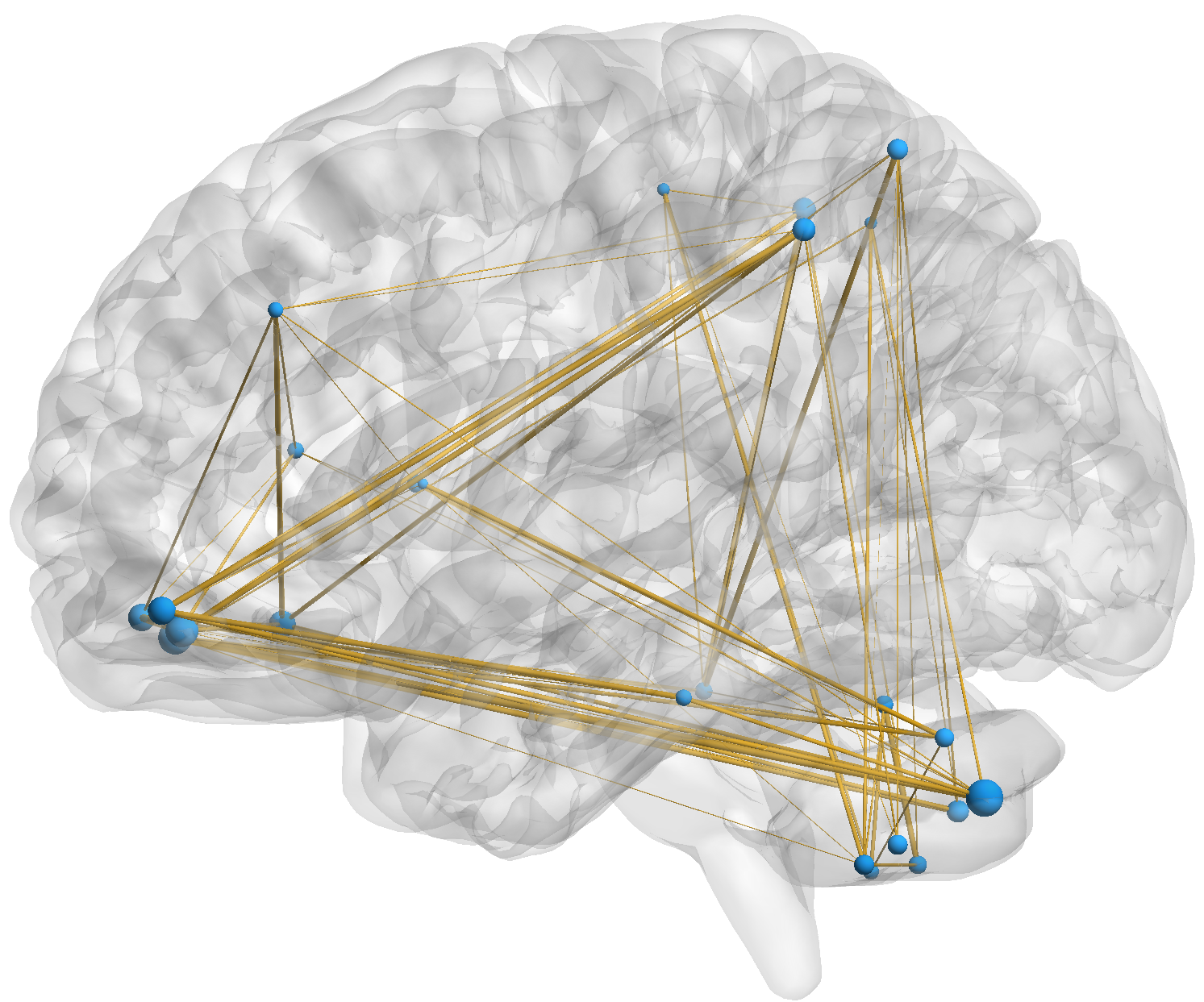}\\
			\hspace{-4mm}\includegraphics[width= \linewidth]{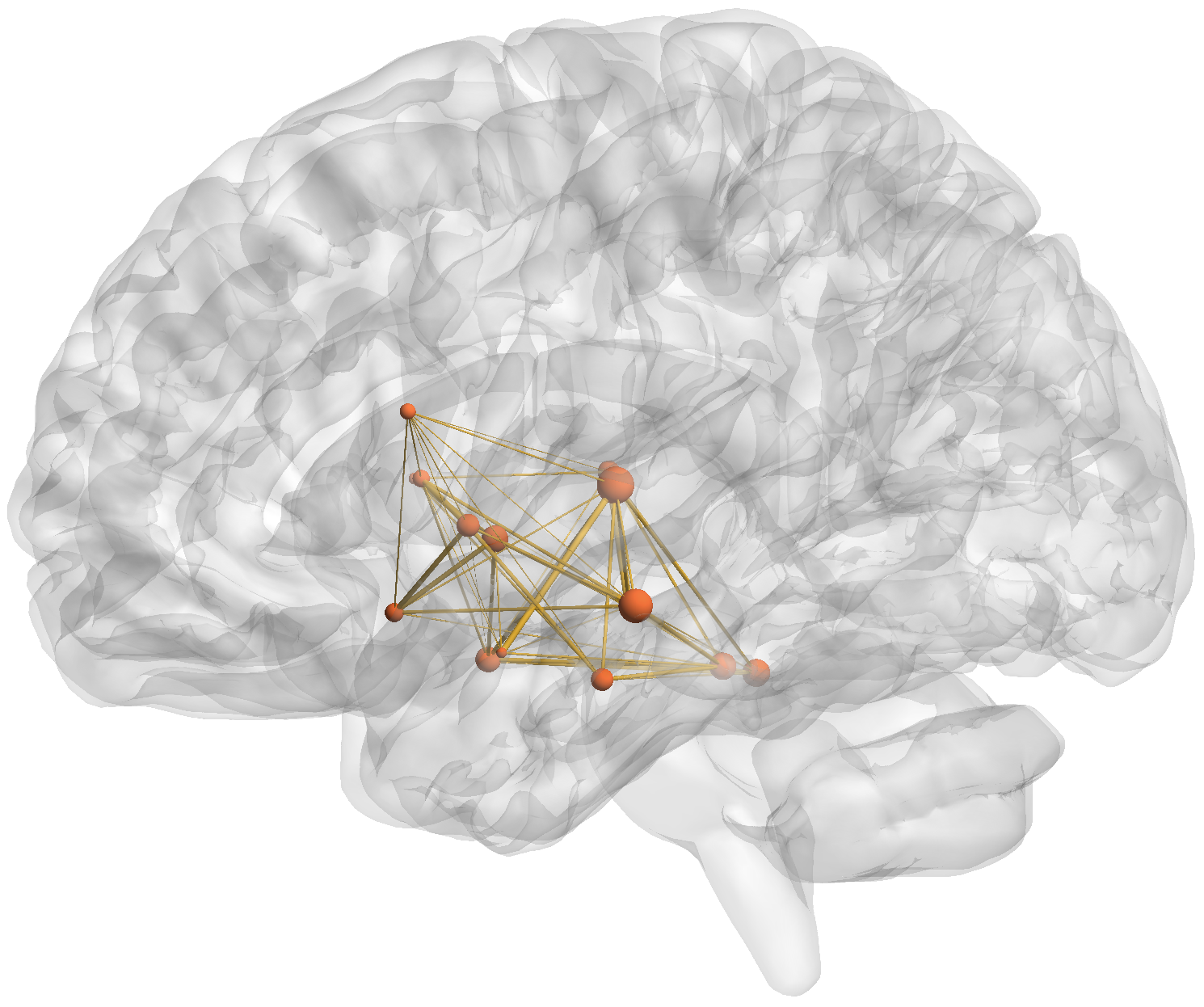}
		\end{tabular}
	\end{minipage}
	\begin{minipage}[c]{.6\linewidth}
		\includegraphics[width= \linewidth]{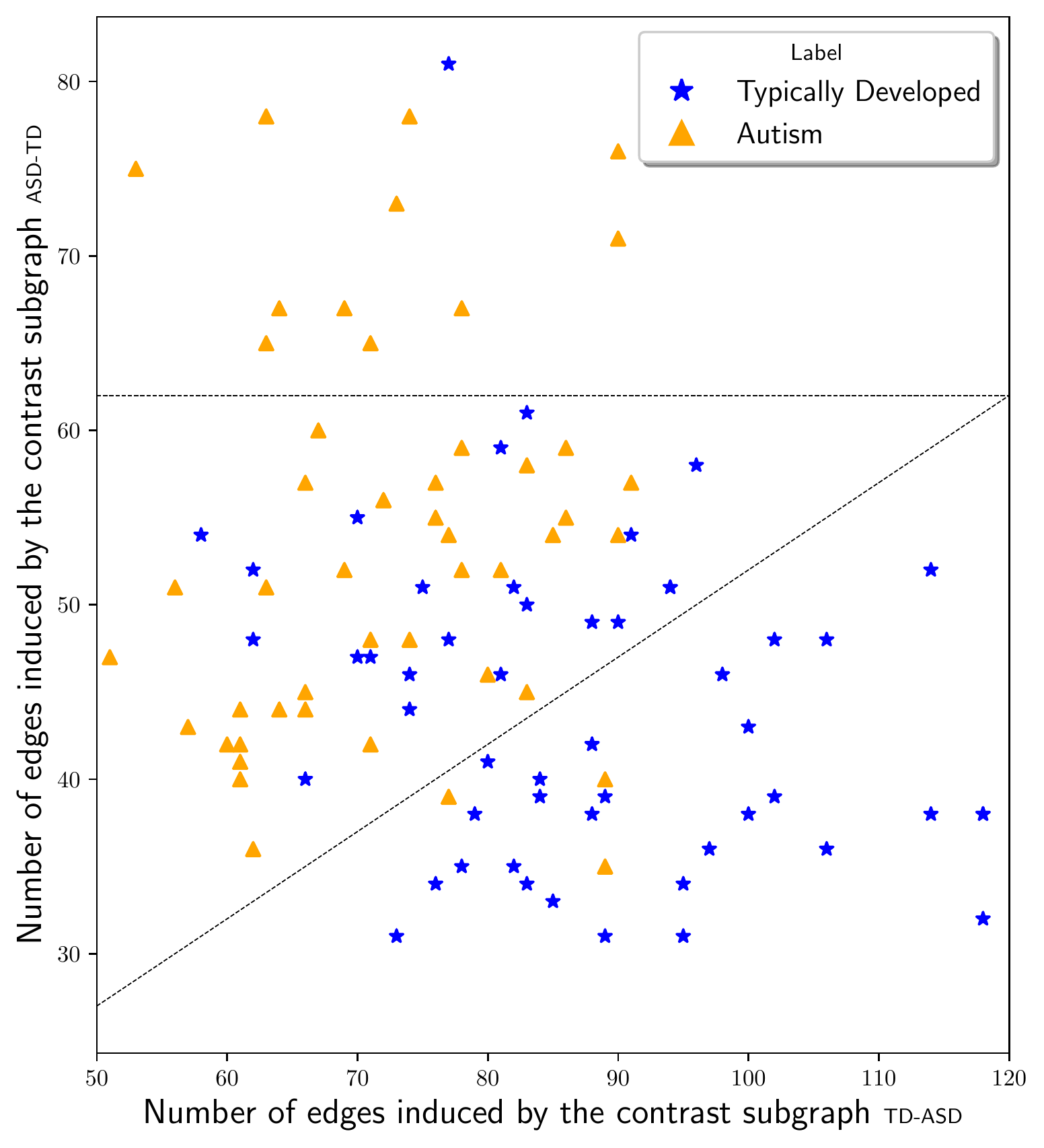}
	\end{minipage}
	\caption{Example of real-world \emph{contrast subgraphs} extracted from a brain-network dataset
	consisting of 49 children affected by Autism Spectrum Disorder (class \asd)
	and 52 Typically Developed (class \td) children.
	Top-left: the contrast subgraph \td-\asd.
	Bottom-left: the contrast subgraph  \asd-\td.
	On the right, scatter plot showing for each individual the number of edges present in the subgraph induced by the contrast subgraph \td-\asd ($x$-axis) and by the contrast subgraph \asd-\td
	($y$-axis). Details are given in Section \ref{subsec:preview} and Section~\ref{sec:experiments_desc}.
	\label{fig:results_large}}
\end{figure}

This task can be seen as a
\emph{graph classification} problem~\cite{wang2017brainclassification,meng2018brainclassification,yan2019groupinn}.
We are given two groups of individuals,
a \emph{condition group}  $\calA$  
and
a \emph{control group} $\calB$. 
Each individual is represented by a brain network, i.e., a graph $G_i=(V,E_i)$,
where each graph is defined over the same set $V$ of vertices
(corresponding to the brain {\roi}s).
The set of edges $E_i$ represents the connections, either structural of functional, between the {\roi}s observed $G_i$.
The goal is to infer a model that, given an unseen graph $G_n=(V,E_n)$,
predicts whether it belongs to class $\calA$ or $\calB$.

Literature on graph classification is mostly based on kernel methods \cite{shervashidze2011weisfeiler},
graph embeddings \cite{narayanan2017graph2vec,adhikari2018sub2vec, gutierrez2019embedding}, and
deep learning \cite{yanardag2015deepWL, lee2018graphclass, wu2019demonet}.
The bulk of these methods for graph classification (with few exceptions discussed in Section \ref{sec:related}), however, is not well-suited for the task of classification in brain networks, for the following reasons.

\begin{description}
  \item[Node-identity awareness:] existing approaches are designed to find structural similarities among the input graphs, without taking into consideration the \emph{correspondence} of nodes among the different networks,
i.e., the fact that a specific vertex id corresponds to the same brain {\roi} in all the input networks.
Not in every application domain it is possible to identify corresponding nodes among
different graphs.
If such a property holds, however, it is crucial to take it into account ---
just ignoring it represents a fatal loss of information.

\smallskip

  \item[Black-box effect:] the inferred models are complex and difficult to understand and to explain
why a certain prediction is made.
\emph{In this specific application domain, the simplicity and explainability
of the models are of uttermost importance.}
In fact, the neuro\-scientist needs to understand which are the {\roi}s and their interactions
that best discriminate between patients and healthy individuals.

\smallskip

  \item[High number of parameters:] a third important limitation is the very high number of parameters that need to be fine tuned,
making the existing approaches too complex to be adopted by non-experts,
especially in application domains where the number of examples is inherently small%
\footnote{Each data point requires an individual undergoing through a scan,
which is an expensive operation.}
and overfitting can be hard to avoid.
\end{description}

\subsection{Our proposal and contributions}
We propose a simple and elegant, yet  effective model,
based on extracting \emph{contrast subgraph},
i.e., a set of vertices whose induced subgraph is dense in the graphs of class $\calA$ and
sparse in the graphs of class $\calB$. 
Our model is extremely simple,
it has only one parameter (governing the complexity of the model), and it has excellent interpretability.
Although it is not the main assessment criteria,
our model also exhibits very good classification accuracy
in several different brain-classification tasks, outperforming more complex models which have several parameters and long training time.

Our main contributions can be summarized as follows:
\begin{itemize}

\item We introduce a novel problem aimed at extracting contrast subgraphs from two classes of graphs,
defined on the same set of vertices (Section \ref{sec:problem}).
We study the complexity of the problem and propose algorithmic solutions (Section \ref{sec:algorithms}).

\smallskip

\item
We apply our method to a brain-network dataset
consisting of children affected by Autism Spectrum Disorder and
children Typically Developed (Section \ref{sec:experiments_desc}).
The analysis of this dataset confirms the interestingness, simplicity and high explainability of the patterns extracted by our methods,
which match domain knowledge in the neuro\-science literature.

\smallskip

\item
We further assess the classification performance of our method
by comparing it against state-of-the-art methods,
on several brain-classification tasks.
The results show that our method, despite its simplicity,
outperforms those more complex methods.

\end{itemize}

We next provide a preview of the type of patterns that our approach can extract from a real-world dataset.

\subsection{A preview of the results}
\label{subsec:preview}
A preview of the type of structures and classification capabilities of our method
is presented in Figure \ref{fig:results_large}.
The data are obtained from the Autism Brain Imagine Data Exchange (ABIDE) project~\cite{craddock2013abide};
more details about data 
are provided in Section~\ref{sec:experiments_desc}.
The specific dataset contains 101 brain networks:
49 patients affected by \emph{Autism Spectrum Disorder} (class \asd) and
52 \emph{Typically Developed} (class \td) individuals.
Each individual is represented by an undirected unweighted graph defined over 116 vertices,
each corresponding to a {\roi}.

On the top-left of Figure \ref{fig:results_large}
we illustrate the contrast subgraph \td-\asd
(the subgraph that maximizes the difference between the number of edges
in the class \td with respect to the same measure for class \asd);
on the bottom-left we illustrate the contrast subgraph \asd-\td.
Vertex size represents the importance of the vertex in discriminating the two classes.
Additionally, the figure depicts the edges whose discrimination capability is above a given threshold.%
\footnote{As we explain in the next section, a contrast subgraph is defined exclusively by a set of vertices. However, in these illustrations we also highlight the most relevant edges.}

By inspecting the two subgraphs,
we can observe an evident complementarity:
\td-\asd involves mainly connections between cerebellum, prefrontal cortex, posterior parietal cortex,
inferior and middle temporal gyri, while
\asd-\td only exhibits connections between striatum and limbic cortex.
These first-sight findings are consistent with recent literature~\cite{khan2015cerebellum,dimartino2011striatum}.
Furthermore, the scatter plot on the right-hand side of Figure \ref{fig:results_large} reports
for each individual (or graph) in the dataset,
the number of edges present in the contrast subgraph \td-\asd ($x$-axis)
and in the contrast subgraph \asd-\td ($y$-axis).
We can see that these very simple and explainable measures, already provide a good separability of the two classes.
For instance, the horizontal dotted line represents the following~rule:

\bigskip
\begin{mdframed}[innerbottommargin=3pt,innertopmargin=3pt,innerleftmargin=6pt,innerrightmargin=6pt,backgroundcolor=gray!10,roundcorner=10pt]
\begin{quote}
	\emph{If an individual exhibits more than 62 edges among the 15 vertices of the contrast subgraph \asd-\td,
	then there are high chances 
	that the individual is affected by \asd.}
\end{quote}
\end{mdframed}
\medskip


\noindent The diagonal dotted line corresponds to another simple rule:

\bigskip
\begin{mdframed}[innerbottommargin=3pt,innertopmargin=3pt,innerleftmargin=6pt,innerrightmargin=6pt,backgroundcolor=gray!10,roundcorner=10pt]
\begin{quote}
\emph{If the number of edges induced by the contrast subgraph \asd-\td is smaller than half of the number of edges induced by the contrast subgraph \td-\asd, then there are high chances 
that the individual is not affected by \asd.}
\end{quote}
\end{mdframed}

Simple rules like these are self-explainable and easily communicable to the neuro\-scientists
for further investigation.

%% file: related.tex
\spara{Graph classification.} Graph classification, i.e., the task of building a model able to predict the target class for unseen graphs accurately, is receiving increasing attention as witnessed by the many approaches proposed in the last few years in the literature.

Shervashidze et al.~\cite{shervashidze2011weisfeiler}, propose a graph kernel based on the Weisfeiler-Lehman test of graph isomorphism. Narayanan et al.~\cite{narayanan2017graph2vec} propose \textit{graph2vec}, a method that considers rooted subgraphs as the components that define a specific graph and performs an embedding under this assumption. Adhikari et al.~\cite{adhikari2018sub2vec} propose
\textit{sub2vec}, whose aim is to learn subgraph embeddings, preserving structural and neighborhood properties.
In our experimental comparison in Section \ref{sec:experiments_desc} we compare the classification accuracy obtained by our method against state-of-the-art baselines \cite{shervashidze2011weisfeiler,narayanan2017graph2vec,adhikari2018sub2vec}. As previously discussed, these methods are not node-identity aware: they do not take into consideration the correspondence of nodes in the different networks. Not in every application domain it is possible to identify corresponding nodes among different graphs, but in application domains (such as brain network classification) in which this is possible, not taking it in consideration is an important limitation. To the best of our knowledge, the only unsupervised embedding for networks with node identity awareness was proposed by
Gutierrez et al.~\cite{gutierrez2019embedding}: we also compare with this method in our experiments.

 If the task of graph classification is something that arose recently, application of it to the domain of brain networks is a growing line of research (see \cite{du2018brainclassificationsurvey} for a comprehensive survey). Recently several researchers~\cite{misman2019brainclassification,meng2018brainclassification, wang2017brainclassification} have proposed the use of deep learning architectures for the same problem we tackle in this paper, i.e., brain networks classification. However, these methods suffer from the black-box curse as they have very little interpretability. Moreover,  they have very high number of parameters that need to be fine tuned and long training time.
A recent work by Yan et al.~\cite{yan2019groupinn}, still based on neural networks, also focuses on explainability. Their proposed architecture has a node-grouping layer before the convolutional layer: it is claimed that such node-grouping layer might provide insight in the most predictive brain subnetworks  w.r.t. the classification task at hand. Direct comparison with~\cite{yan2019groupinn} is not doable as their pipeline starts directly from the timeseries associated to the ROIs, and early on splits the computation in two tracks: the positive and the negative correlation track. Instead, the input to the problem studied in this paper, is simply two groups of graphs.


\spara{Contrast and dense subgraphs.} Contrast, sometimes called discriminative, subgraphs have been approached in the graph mining literature mostly as a \emph{frequent subgraph mining} task \cite{Bailey06,LTS,CORK,LEAP}: in this context the discriminative structure is a typically small motif which is defined over node-labelled graphs (e.g., a triangle \textsf{h-h-o}), that appears frequently in one set of graphs but infrequently in another set of graphs. Frequent subgraph discovery usually requires to enumerate all the subgraphs satisfying the search criteria, and  it is more appropriate for applications in bioinformatics or cheminformatics dealing with large databases of proteins, molecules or chemical compounds. Counting frequency requires to solve many subgraph isomorphism instances, and as such it is computationally challenging.
This part of the literature is rather far from the problem studied in this paper, as our structure of interest is a set of vertices whose induced subgraph is dense in a class of graphs and sparse in another: we do not have any notion of frequency, nor we require node-labels, nor need to solve subgraph isomorphism.

Closer to our work is, instead, the literature on dense-subgraph discovery \cite{gionis2015dense}.
Among the many notions of \emph{density} studied in the literature, the most popular one
is the \emph{average-degree density}, as the problem of extracting the subgraph maximizing such density (commonly referred to as the \emph{densest subgraph}) is solvable in polynomial time~\cite{Goldberg84} and admits a linear time $\frac{1}{2}$-approximation algorithm~\cite{AITT00}.
As discussed in Section \ref{sec:algorithms} an alternative notion of dense subgraph is that of optimal quasi-clique introduced by Tsourakakis et al.~\cite{tsourakakis2013denser}, which has been shown to return denser subgraphs and with smaller diameter than the so-called densest subgraph.

Several recent works have dealt with the problem of extracting dense subgraphs from a set of multiple graphs sharing the same vertex set \cite{PeiJZ05,jiang2009mining,boden2012mining,WuZLFJZ16,galimberti2017core}, however none of these deal with extracting contrast or discriminative subgraphs.
The work probably most related to ours is  due to Yang et al.~\cite{YangCZWPC18}, which study the density contrast subgraph, i.e., the set of vertices whose maximize the difference between the density of its induced subgraphs in two input graphs $G^A$ and $G^B$. The work of Yang et al.~\cite{YangCZWPC18} differs from ours as they adopt, as measure of density,  the \emph{average degree}~\cite{Goldberg84} and \emph{graph affinity} \cite{affinity}. They show that their problem is equivalent to the densest subgraph problem with negative edges, which is shown to be  \NPhard and hard to approximate.
In our work, instead of adopting the densest subgraph definition --- as done by Yang et al.~\cite{YangCZWPC18} --- which normalizes by the size of the subgraph $S$, we follow Tsourakakis et al.~\cite{tsourakakis2013denser} and we balance the total edge weight with the term $\alpha \binom{|S|}{2}$. The parameter $\alpha$ also allows us to control the size of the extracted contrast subgraphs. Finally, Yang et al.~\cite{YangCZWPC18} do not consider the case in which the input is constituted by many graphs in two classes and its potential application in graph classification.


%% file: problem.tex
\spara{Notation.}
We consider a dataset \dataset of observations,
where the $i$-th observation corresponds to a graph $G_i=(V,E_i)$.
Without loss of generality we assume that
all observation graphs are defined over the same set of vertices $V$,%
\footnote{When the observation graphs have different set of vertices, we can consider that they are defined on the union of all vertex sets.}
i.e., brain regions of interest,
while the edge set $E_i$ represents connections between vertices in the observation graph $G_i$.
For sake of simplicity of presentation we consider each $G_i$ to be unweighted, even if the model can straightforwardly deal with weighted input graphs.
The dataset \dataset is divided in two groups:
the \emph{condition group} $\calA=\{G_1^\A,\ldots,G_{r_\A}^\A\}$ and
the \emph{control group} $\calB=\{G_1^\B,\ldots,G_{r_\B}^\B\}$.

We aggregate the information in the groups \calA and \calB
in two \emph{summary graphs} $\GA=(V,\wA)$ and $\GB=(V,\wB)$, respectively.
The summary graphs \GA and \GB are undirected and weighted graphs,
defined over the same set of vertices $V$ as the original observation graphs,
while
$\wA, \wB: V \times V \rightarrow \realsnn$
are weight functions assigning a value to each pair of vertices
and summarizing the edges of the observation graphs in the groups \calA and \calB.

In one particular instantiation, given two vertices $u$ and $v$ in $V$,
we define $\wA(u,v)$ to be the \emph{fraction} of graphs $G_i^\A \in \calA$
in which $u$ is incident to $v$,
that is,
\begin{equation}\label{eq:summary}
\wA(u,v) = \frac{1}{r_{\A}} \left| G_i^\A \in \calA \text{ s.t. }  (u,v)\in E_i^\A \right|,
\end{equation}
and similarly for \wB.
Note that according to this weighting function $\wA(u,v) \in  [0,1]$ with $\wA(u,v) = 0$ denoting the case in which there is no relationship
(i.e., no edge) between $u$ and $v$ in~$\GA$.
Alternative definitions of the weight functions \wA and~\wB are possible.
For instance, if the input graphs are all weighted, with $w^A_i(u,v)$ denoting the weight of the edge $(u,v)$ in $G_i^\A$, then we can define:
\[
\wA(u,v) = \frac{1}{r_{\A}} \sum_{G_i^\A \in \calA} w^A_i(u,v).
\]

As another example, one could require the summary graph itself to be binary, i.e., $\wA(u,v) \in  \{0,1\}$,
for instance by considering if the number of edges present in the observation graphs
is larger than a given threshold.


Given a subset of vertices $S\subseteq V$,
we consider the subgraph induced by $S$ in the summary graphs  \GA and \GB.
We define
\[
\eA(S) = \sum_{u,v \in S} \wA(u,v)
\]
the sum of edge weights in the subgraph of \GA induced by the vertex set $S$.
In the case of summary graphs with binary weights,
the quantity $\eA(S)$
corresponds to the number of edges in the subgraph induced by $S$ in the summary graph \GA.
Analogous definitions apply to \GB.

\spara{Problem definition.}
The basic definition of the contrast subgraph problem requires to find a subset of vertices whose induced subgraph is dense in a summary graph \GA and sparse in summary graph \GB.

\begin{problem}[Contrast subgraph]
\label{prb:contrast}
Given two sets of observation graphs,
i.e.,
the condition group $\calA=\{G_1^\A,\ldots,G_{r_\A}^\A\}$ and
the control group $\calB=\{G_1^\B,\ldots,G_{r_\B}^\B\}$,
and corresponding summary graphs $\GA = (V, \wA)$ and $\GB = (V, \wB)$,
we seek to find a subset of vertices $S^* \subseteq V$ that maximizes the
\emph{contrast-subgraph objective}
\begin{eqnarray*}
\delta(S)
 & = & \eA(S) - \eB(S) - \alpha \binom{|S|}{2} \\
 & = & \sum_{u,v \in S} \left( \wA(u,v) - \wB(u,v) - \alpha \right),
\end{eqnarray*}
where $\alpha \in \realsnn$ is a user-defined parameter.
\end{problem}

The last term of the objective, i.e., $- \alpha \binom{|S|}{2}$,
is a regularization term penalizing solutions of large size:
the larger the value of $\alpha$, the smaller is the optimal contrast subgraph.
Note that, to avoid the na\"ive solution, $S^* = \emptyset$,
we have to ensure that $$0 < \alpha < \max_{u,v \in V} \left(\wA(u,v) - \wB(u,v)\right),$$
otherwise we would encounter the case in which every pair of vertices is detrimental
for the objective function.

\begin{figure}[t]
  \centering
  \includegraphics[width=.7\linewidth]{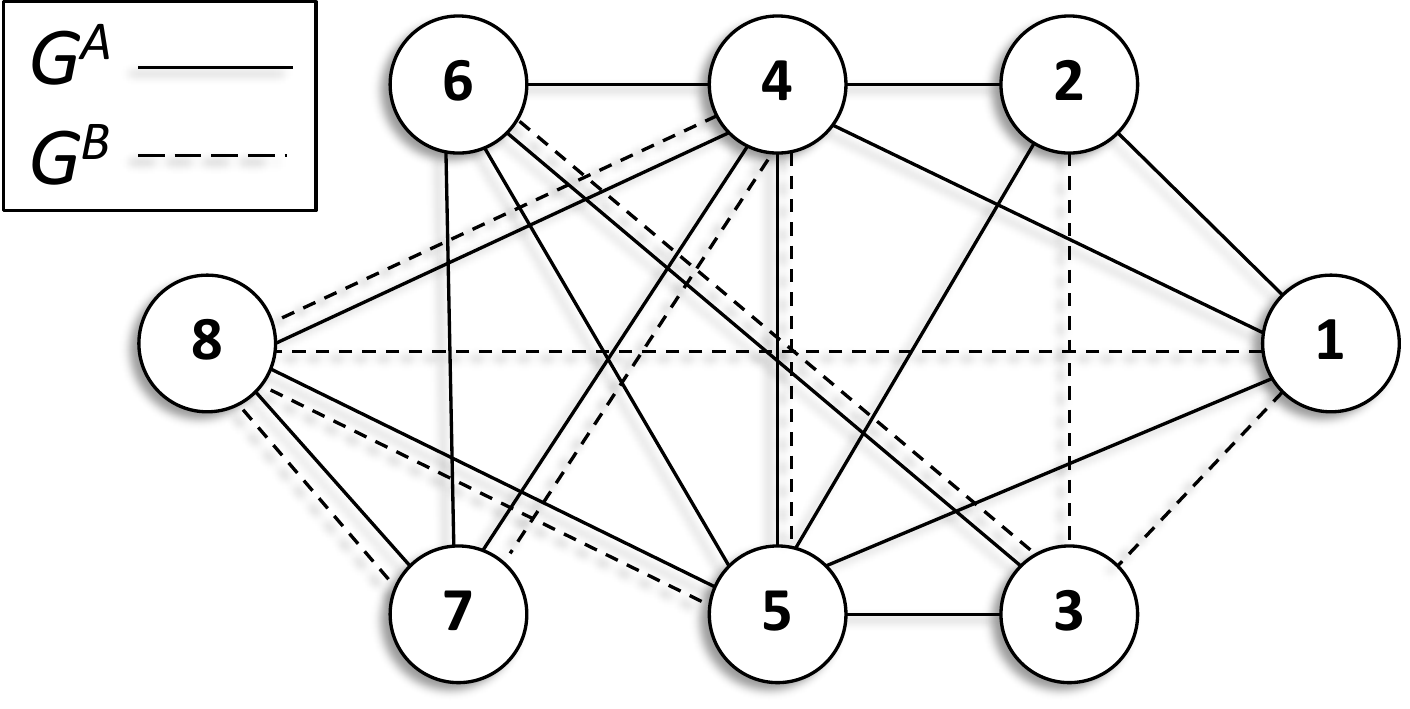}
  \caption{Run-through example of two input graphs \GA and \GB. Please refer to Example \ref{example1},  \ref{example2},  \ref{example3}, and  \ref{example4} for different definitions of contrast subgraphs.
  }\label{fig:example1}
\end{figure}

\begin{myexample}
\label{example1}
Figure \ref{fig:example1} provides an example of two summary graphs \GA and \GB.
In this example we assume that the summary graphs are unweighted, i.e., the edge weight function is binary.
For $\alpha = 0.8$ the contrast subgraph distinguishing \GA from \GB is given by the set of vertices $S_3=\{1,2,4\}$.
In particular, $S_3$ forms a clique in \GA 
and an independent set in \GB. 
The contribution of each edge to the objective function is $1 - 0 - 0.8 = 0.2$,
giving a total value for the objective $\delta(S_3)=0.6$.
Note that the set $S_4=\{1,2,4,5\}$ is a larger clique in \GA,
but it induces an edge in \GB.
In this case, the contribution to the solution $S_4$ is $0.2$ for all $5$ edges appearing only in \GA,
and $-0.8$ for the edge appearing in both \GA and \GB,
giving a total value for the objective $\delta(S_4)=0.2$.
\end{myexample}
\begin{myexample}
\label{example2}
In the example of Figure \ref{fig:example1}, with $\alpha = 0.5$,
the contrast subgraph distinguishing \GA from \GB is given by the set of vertices $S_5=\{1,2,4,5,6\}$.
In this case, the contribution in the objective of each edge appearing only to \GA (there are 7) is
$1 - 0 - 0.5 = 0.5$, while an edge appearing in both graphs (there is only one edge) contributes $-0.5$,
and the same for an edge that does not appear in none graph (there are 2 edges).
Finally, an edge appearing only in \GB (there are no such edge in the subgraph induced by $S_5$)
would contribute $-1.5$.
Thus, the total value of the objective is $\delta(S_5)=2$.
Note that $S_4=\{1,2,4,5\}$ achieves the same score.
\end{myexample}


The last two examples show the function of the parameter $\alpha$, which is that of governing the complexity of the extracted patterns.
As any other parameter governing the complexity of a model, the ``best'' value of $\alpha$ might depend on the specific input dataset, as well as on trading-off different requirements of the analyst: for instance, when explainability is important, a smaller contrast subgraph (larger $\alpha$) might be desirable, but if it is too small might end up being not interesting.

So far we have asked to find the contrast subgraph that distinguishes \GA from \GB.
Note, however, that our objective function is not symmetric with respect to \GA and \GB.
Thus, one can consider the same problem for finding the contrast subgraph that distinguishes \GB from \GA.
Indeed, in our experiments in Section \ref{sec:experiments_desc} we always consider both contrast subgraphs.


\begin{myexample}
\label{example3}
In the example of Figure \ref{fig:example1}, with $\alpha = 0.8$,
the optimal contrast subgraph that distinguishes \GB from \GA is a single edge appearing in \GB and not in \GA.
Thus, any of the following sets is an optimal solution: $\{1,8\}$, $\{1,3\}$ or $\{2,3\}$.
\end{myexample}

\spara{Symmetric variant.}
In some cases
the analyst  may want to find a subgraph having the largest absolute difference, in terms of edge weights,
between \GA and \GB, disregarding on whether the larger weights are on one side or the other.
To address this case, we consider a symmetric variant of Problem~\ref{prb:contrast}.

\begin{problem}[Symmetric contrast subgraph]
\label{prb:contrast_variant}
Given two sets of observation graphs,
i.e.,
the condition group $\calA=\{G_1^\A,\ldots,G_{r_\A}^\A\}$ and
the control group $\calB=\{G_1^\B,\ldots,G_{r_\B}^\B\}$,
and corresponding summary graphs $\GA = (V, \wA)$ and $\GB = (V, \wB)$,
we seek to find a subset of vertices $S^* \subseteq V$ that maximizes the
\emph{contrast-subgraph objective}
\begin{eqnarray*}
\sigma(S)
 & = & \sum_{u,v \in S} \left( \left| \wA(u,v) - \wB(u,v) \right| - \alpha \right),
\end{eqnarray*}
where $\alpha \in \realsnn$ is a user-defined parameter.
\end{problem}


\begin{myexample}
\label{example4}
In Figure \ref{fig:example1}
the optimal symmetric contrast subgraph (Problem \ref{prb:contrast_variant}), for $\alpha = 0.5$, is $\{1,2,3,5\}$.
This vertex set induces a clique in the union of the two summary graphs, such that each edge belongs only to either \GA or \GB, so that every pair of vertices produces a positive contribution to the objective function of Problem \ref{prb:contrast_variant}.
\end{myexample}

%% file: algorithms.tex

To tackle the problem of extracting a dense subgraph from a given unweighted graph,
Tsourakakis et al.~\cite{tsourakakis2013denser} introduce the notion of \emph{edge-surplus}
$f_\alpha(S)$ for a subgraph induced by a set of vertices $S \subseteq V$,
which they defined as follows
\begin{equation}\label{eq:surplus}
f_\alpha(S) = \begin{cases}
0, & S = \emptyset ;\\
g(e[S]) - \alpha h(|S|), & otherwise,
\end{cases}
\end{equation}
where $e[S]$ denotes the number of edges in the induced subgraph,
$\alpha > 0$ is a fixed penalty parameter and $g$ and $h$ are two strictly-increasing functions.
The rationale behind the definition of edge surplus is to counterbalance two opposite forces:
the term $g(e[S])$ favors subgraphs with many edges, whereas the term $\alpha h(|S|)$ penalizes large subgraphs.
They propose the problem of extracting the subgraph with maximum edge surplus
as a general class of dense-subgraphs problems.

Tsourakakis et al.~\cite{tsourakakis2013denser} then focus on a specific instance of the general problem,
i.e., $g(x) = x$ and $h(x) = \binom{x}{2}$, which they dub as \textsc{Optimal Quasi-Clique} (\oqc) problem.
They propose an approximation algorithm for the \oqc problem, which was later shown to be \NPhard in \cite{tsourakakis2015}.

More recently, Cadena et al.~\cite{cadena2016dense} extend the \oqc setting,
and consider extracting dense subgraphs in \emph{streaming signed networks for event detection}.
They generalize \oqc to weighted graphs and edge-dependent penalty parameter $\alpha$:

\begin{problem}[Generalized optimal quasi-clique (\goqc) \cite{cadena2016dense}]
\label{prb:edsn}
Given a graph $G=(V,E)$, and functions $w(u,v)$ and $\alpha(u,v)$,
for each pair of vertices $u,v \in V$,
find a subset of vertices $S \subseteq V$ that maximizes
\[
f(S) = \sum_{u,v \in S} w(u,v) - \alpha(u,v).
\]
\end{problem}

Cadena et al.\ prove that Problem \ref{prb:edsn} is \NPcomplete and
\NPhard to approximate within a factor $\bigO(|V|^{1/2 - \epsilon})$ \cite[Theorem~1]{cadena2016dense}.
Then they develop an algorithm using a semidefinite-programming (SDP) based rounding to produce a solution,
which is then refined by the \emph{local-search} procedure of Tsourakakis et al.~\cite{tsourakakis2013denser}.
Their algorithm provides a $\bigO(\log n)$ approximation guarantee,
although in practice the approximation is shown to be much better.

We next show that, although defined in a totally different setting (event detection in a single, streaming, signed network), we can fruitfully make use of the algorithm developed by Cadena et al.\
to solve the contrast subgraph problems we defined in the previous section.
Next proposition provides the mapping between our problems and the \goqc problem.

\begin{proposition}
\label{fact:mapping}
Problems \ref{prb:contrast} and~\ref{prb:contrast_variant}
can be mapped to Problem~\ref{prb:edsn}.
\end{proposition}
\begin{proof}
The mapping is given by setting
$\alpha(u,v) = \alpha$ and
\[
w(u,v)
 =
\begin{cases}
\wA(u,v) - \wB(u,v) \quad\quad& \text{(Problem}~\ref{prb:contrast}\text{)} \\
|\wA(u,v) - \wB(u,v)|, & \text{(Problem}~\ref{prb:contrast_variant}\text{)}
\end{cases}
\]
for each pair of vertices $u,v \in V$.
\end{proof}

Based on this proposition, the algorithms that we use for our problems are reported in details in Appendix~\ref{appendix:pseudocode}.

%% file: experiments.tex
The aim of our experimental evaluation is to show how contrast subgraphs
can be exploited profitably for finding discriminative patterns between two groups of brain networks.
In particular, our goal is to answer the following questions:

\begin{enumerate}
\item
Can the contrast-subgraph approach be used to identify structural differences between two groups of brain-networks,
which are not easy to detect through standard analysis?
\item
Are the discovered contrast subgraphs interpretable?
\item
Are the results obtained with the contrast-subgraph approach consistent with the neuroscience literature?
\item
Does the approach achieve good performance for the classification tasks considered?
\end{enumerate}

\spara{fMRI and the brain network.}
We start by providing some background on fMRI and brain networks.
The analysis reported in this section are based on \emph{resting-state functional MRI} (rs-fMRI) data.
The rs-fMRI is a technique whose aim is to measure brain activity
in the absence of any underlying controlled experimental paradigm,
exploiting the link between neural activity and blood flow and oxygenation.
When a neuron activates, the blood flow in the corresponding regions increases,
and as a result, oxygen-rich blood displaces oxygen-depleted blood.
Such a variation is measured as a signal, dubbed BOLD (blood-oxygen-level dependent),
and it is what really influences the final magnetic resonance signal.
Biswal et al. \cite{biswal1995fmri} were the first to demonstrate temporal coherence
between several blood fluctuation in the primary sensor motor cortices and,
along with several corroborating results,
prompted researchers to think of the brain as a network,
where nodes represent regions and edges represent functional connectivity measured by correlation.
Since the quantity of signals detected by any MRI is huge,
there is usually a process of aggregating voxels,
so as to reduce the data dimensionality.
Furthermore, since such signals are heavily subject to noise caused by different confounding factors,
a plethora of pre-processing strategies,
either from the signal-processing side (e.g., filtering, signal correction, etc.)
or from the network-analysis side (e.g., thresholding)
have been proposed.
We refer the interested reader to the comprehensive survey of Lang et al.~\cite{lang2012preprocessing}.

\begin{table}[t]
\caption{Datasets used in the experiments. \label{tab:datasets}}
\vspace{-2mm}
	\begin{tabular}{cccc}
	\toprule
		\multicolumn{1}{c}{Dataset} & \multicolumn{1}{c}{Description}  & \multicolumn{1}{c}{\td} & \asd \\ \midrule
		\highlight{Children} & Age $\leq 9$  & 52 & 49 \\
		\highlight{Adolescents} & Age in $[15,20]$  & 121 & 116 \\
		\highlight{EyesClosed} & Eyes closed during scanning  & 158 & 136 \\
		\highlight{Male} & Male individuals & 418 & 420 \\
	\bottomrule
	\end{tabular}
\vspace{2mm}
\end{table}

\spara{Data source.}
For our experimental evaluation
we use a publicly-available data\-set\footnote{\url{http://preprocessed-connectomes-project.org/abide/index.html}}
released by the Autism Brain Imagine Data Exchange (ABIDE) project \cite{craddock2013abide}.
The dataset contains neuroimaging data of 1112 different patients,
573 Typically Developed (\td) and
539 suffering from Autism Spectrum Disorder (\asd).

To extract the data used in this work,
we have followed the preprocessing strategy denoted as DPARSF,%
\footnote{\url{http://preprocessed-connectomes-project.org/abide/dparsf.html}}
followed by Band-Pass Filtering and Global Signal Regression.
To parcellate the brain we adopt the AAL atlas \cite{tzouriomayer2002aal},
which divides the brain into 116 {\roi}s.
The final result of our pre-processing, for a single patient,
is a set of 116 time series (each one associated to a \roi) of length 145.
In order to obtain the brain connectome we compute the pairwise Pearson correlation
between the time series of each pair of {\roi}s,
producing a 116$\times$116 correlation matrix.
Finally, we use a threshold~$t$, whose value is set equal to the 80-th percentile
of the distribution of correlation values,
and we draw an edge for each pair $(u,v)$ having correlation larger than $t$;
these choices are typical in the literature, see for example
the works of Lord et al.~\cite{lord2012percentile} and Rubinov et al.~\cite{rubinov2009percentile}.

The end result of our data processing is an \emph{undirected unweighted graph} for each patient.

It should be noted that fMRI data, in addition to capturing the status of an individual,
are influenced by other phenotypic information,
which for our purposes would result to confounding factors with respect to our target variable,
i.e., the status \td and \asd.
As an effort to mitigate the intrinsic variance and in accordance with neuroscience literature,
we create four different datasets by selecting individuals who share some common characteristics,
such as age, gender, or the condition during the scan, e.g., eyes closed.

More information about the four selected datasets,
and the sizes of the two target classes,
is reported in Table \ref{tab:datasets}.
We remind that in all the datasets the classes of the observations are \td and \asd, the name we use to identify each dataset is the common phenotypic features shared by the observations.
\highlight{Children} contains only individuals whose age is at most 9 years.
\highlight{Adolescents} contains individuals whose age is between 15 and 20 years.
\highlight{EyesClosed} contains individuals who performed their fMRI with eyes closed.
Finally,
\highlight{Male} contains only male individuals.
For all our datasets we consider only individuals for
whom there are no missing observations in the time series,
and we apply the pre-processing described above.
As a result, each individual is represented by an undirected unweighted graph with $|V| = 116$ vertices.

\begin{figure}[t]
\centering
	\includegraphics[width=\linewidth]{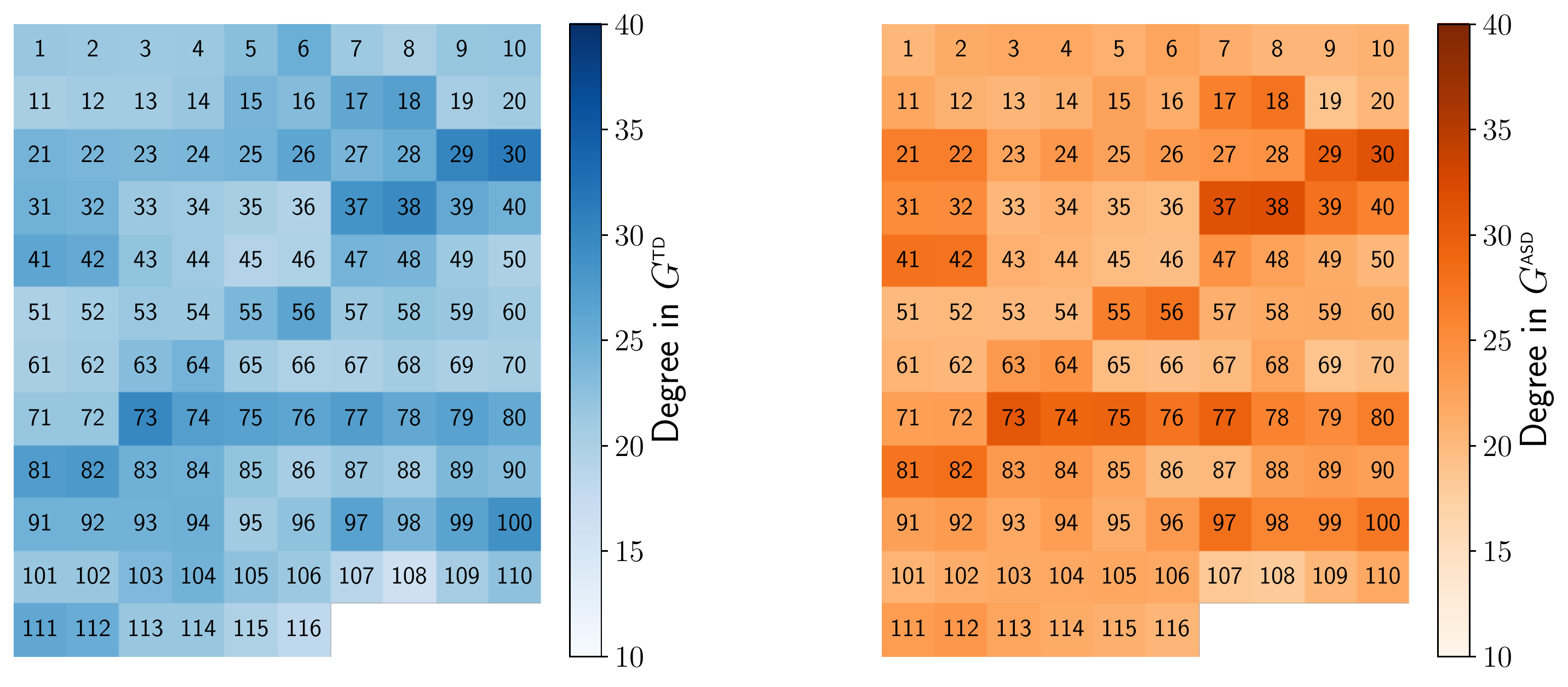}
	\caption{Weighted degree of the 116 nodes (ROIs) in \GTD (left) and \GASD (right). The weighted degree of each vertex is almost identical in the two summary graphs, not revealing any structural difference between the two classes.	\label{fig:deg}}
\end{figure}

\subsection{Characterization}
\label{sec:descriptive}

In this section we characterize the type of information produced by the approach of contrast subgraphs,
focusing on the \highlight{Children} dataset.
A preview of a contrast subgraph for this dataset was already given %
in Figure \ref{fig:results_large} (Section~\ref{subsec:preview}).

The first natural question we consider is whether
the contrast subgraphs capture some simple, first-level information,
which might characterize the two different classes.
To answer this question, we check the two summary graphs
\GTD and \GASD, produced according to Equation (\ref{eq:summary}).
These are two weighted graphs, where edge weights take values in the interval $[0,1]$,
representing the fraction of individuals in the class who exhibit the edge.

In Figure \ref{fig:deg} we report, for each of the 116 {\roi}s, their (weighted) degree in \GTD and \GASD.
The degrees of the nodes in a network is a key feature.
If there was some evident structural difference between \GTD and \GASD,
we would expect to see different with respect to the degrees of their nodes in the two graphs.
Instead, in Figure \ref{fig:deg} we can observe \emph{very similar} degrees for all {\roi}s
in both summary graphs.
In other words, we are not able to identify a few nodes that are ``important''
in \GTD and not in \GASD,
or vice versa.

We thus move our attention to the edge weights in the \emph{difference graph}
$\diffnet{\td}{\asd}= (V, \wTD-\wASD)$,
looking for patterns that are distinct in one network or the other.
The adjacency matrix of $\diffnet{\td}{\asd}$ is reported in \reffig{diff}.%
\footnote{We do not show adjacency matrix of \diffnet{\asd}{\td}
since it would be equal to the one in \reffig{diff} with opposite signs.}
It is particularly interesting to check whether in this adjacency matrix
there are regions with a predominant color.
Such regions could be either single cells or rows/columns,
and they can be either dark red or dark blue.
In the case of single cells, such dark regions would represent
a single edge showing a large weight difference between the two classes.
In the case of rows/columns, such dark regions would represent
{\roi}s with a tendency to be hyper/hypo-connected in a specific class.
However, we observe that none such pattern appears.
Instead, the whole matrix is represented with light-colored cells
(i.e., values are very close to 0),
exhibiting very weak structural difference between the two classes.

\begin{figure}[t]
	\includegraphics[width= .7\linewidth]{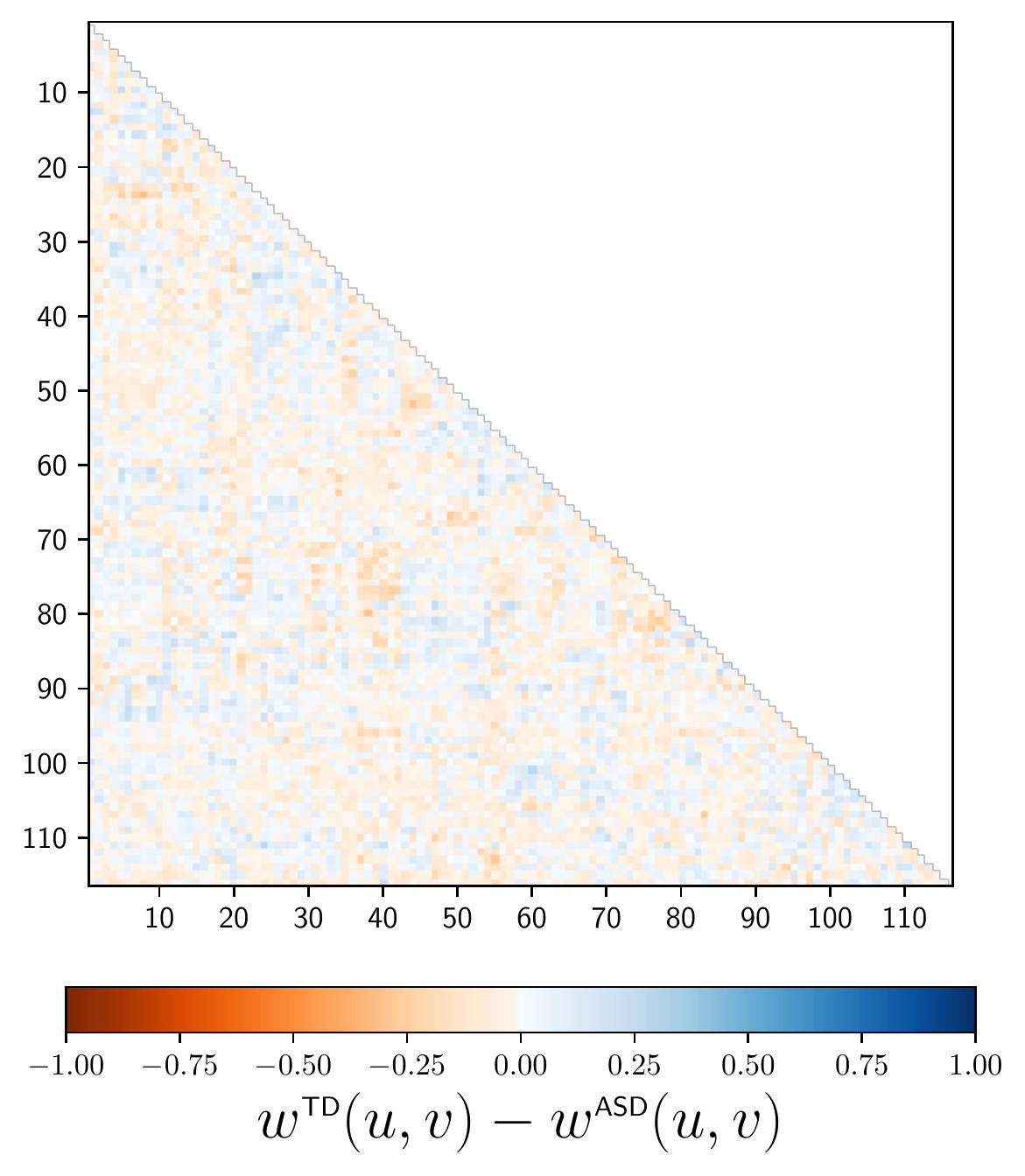}
	\caption{Adjacency matrix of \diffnet{\td}{\asd}.
	Elements of the matrix are colored according \wTD-\wASD.
	Nodes are ordered by their id in the AAL atlas.	\label{fig:diff}}
\end{figure}

\spara{Contrast subgraph (\refpr{contrast}).}
Despite the apparent lack of structural differences between the two classes,
we have seen already in Section \ref{subsec:preview} (Figure \ref{fig:results_large})
an example of a contrast subgraph,
which provides simple discrimination rules between the two classes, \td and \asd.
Another example, obtained using a larger value of $\alpha$ and thus yielding a smaller solution,
is reported in \reffig{results_small}.
In this setting the contrast subgraph \td-\asd (top left)
contains only 7 nodes,
while the contrast subgraph \asd-\td (bottom left) contains only 6 nodes.
As in Figure \ref{fig:results_large},
the size of a node in the contrast subgraph \td-\asd (respectively, \asd-\td)
is proportional to its weighted degree in the \emph{difference graph} $\diffnet{\td}{\asd}$ (respectively, $\diffnet{\asd}{\td}$).
An induced edge $(u,v)$ is depicted in the figure of the contrast subgraph \td-\asd
only if $\wTD-\wASD(u,v) \geq 0.1$ (and similarly for \asd-\td).
Inspection of the data points in \reffig{results_small} motivates us to derive a simple rule:

\medskip
\begin{mdframed}[innerbottommargin=3pt,innertopmargin=3pt,innerleftmargin=6pt,innerrightmargin=6pt,backgroundcolor=gray!10,roundcorner=10pt]
	\begin{quote}
		\emph{If the number of edges induced by the contrast subgraph \asd-\td is smaller than the number of edges induced by the contrast subgraph \td-\asd, then there are high chances that the individual is affected by \asd.}
	\end{quote}
\end{mdframed}
\medskip

\begin{figure}[t!]
	\begin{minipage}[c]{.3\linewidth}
		\begin{tabular}{c}
			\hspace{-4mm}\includegraphics[width= \linewidth]{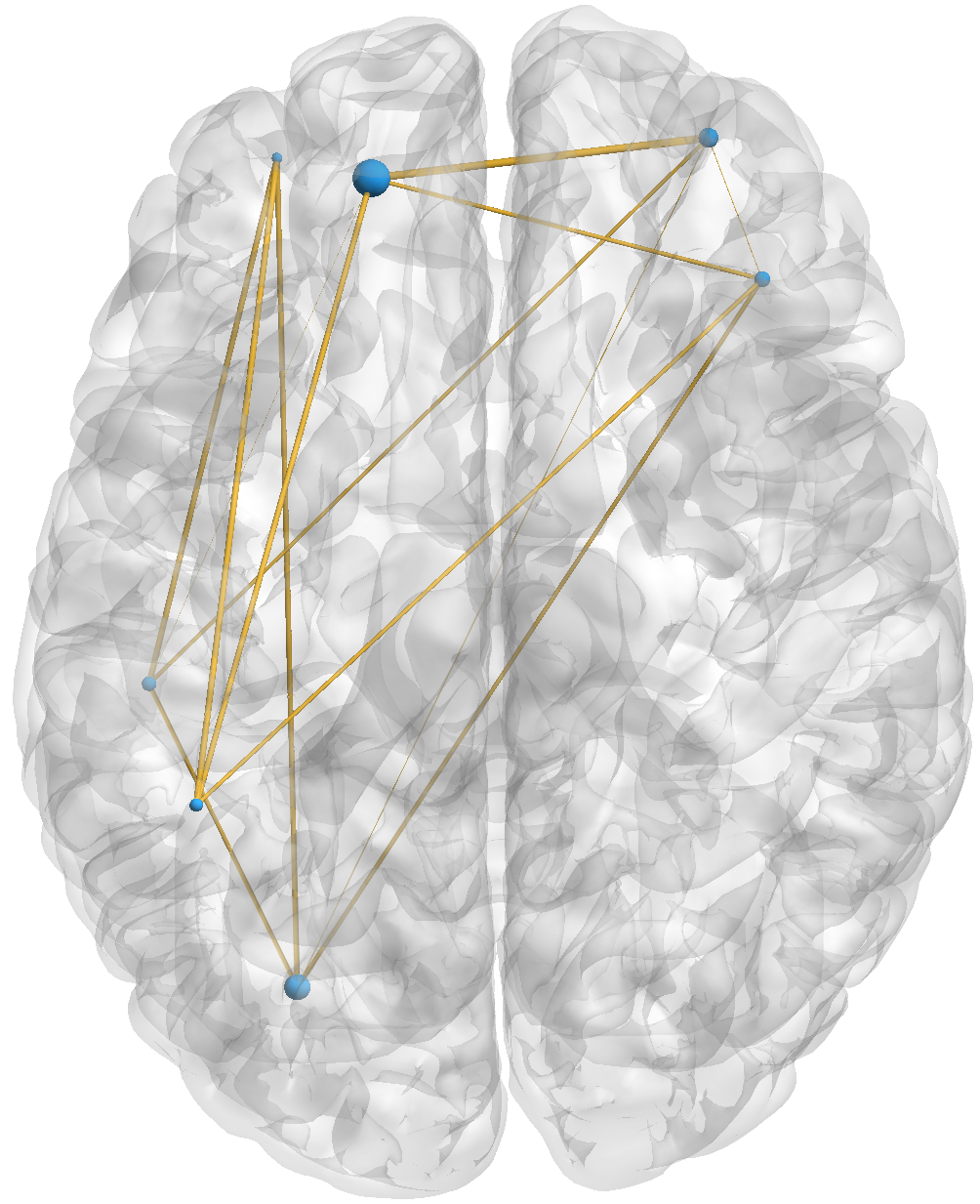}\\
			\hspace{-4mm}\includegraphics[width= \linewidth]{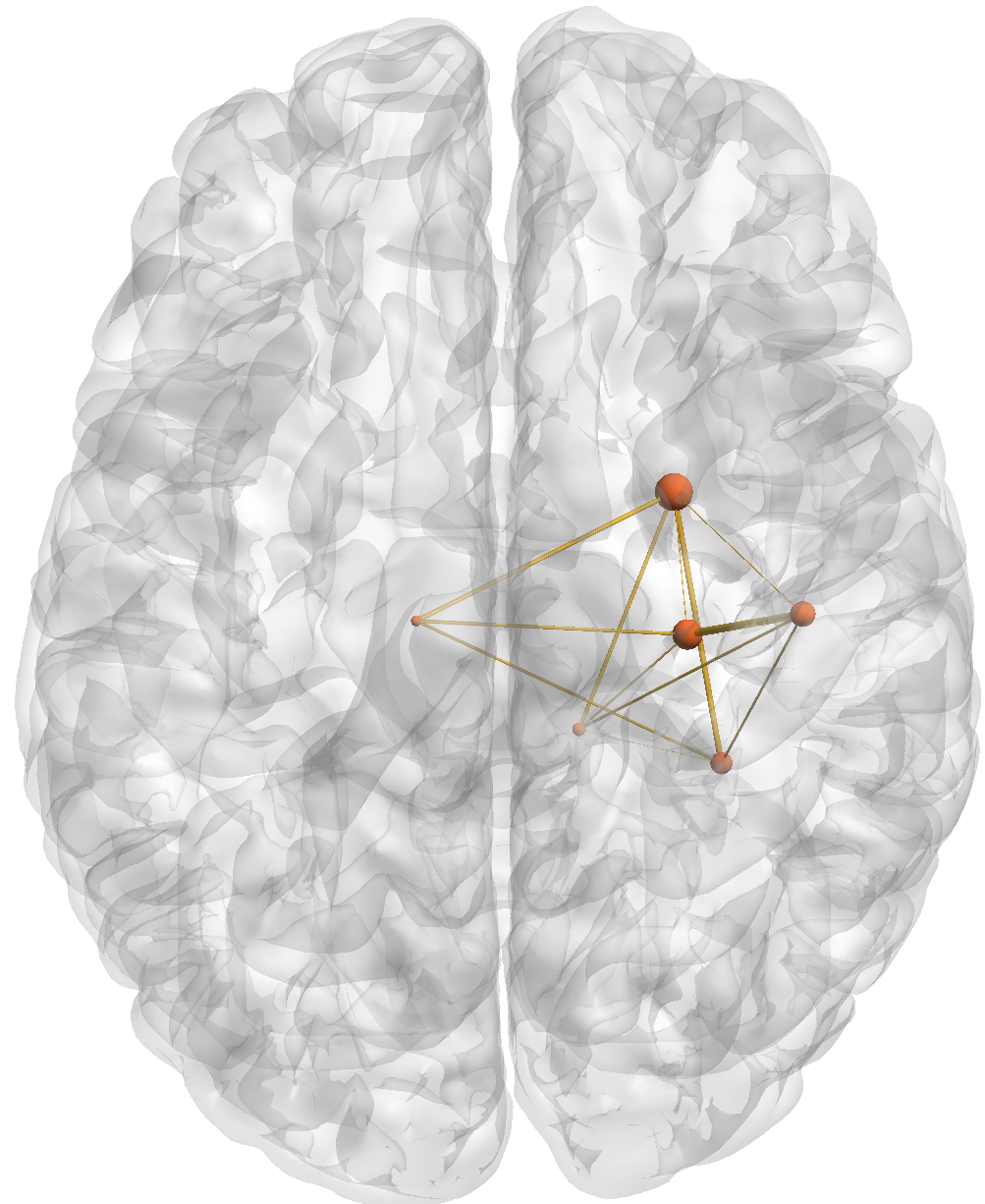}
		\end{tabular}
	\end{minipage}
	\begin{minipage}[c]{.68\linewidth}
		\includegraphics[width= \linewidth]{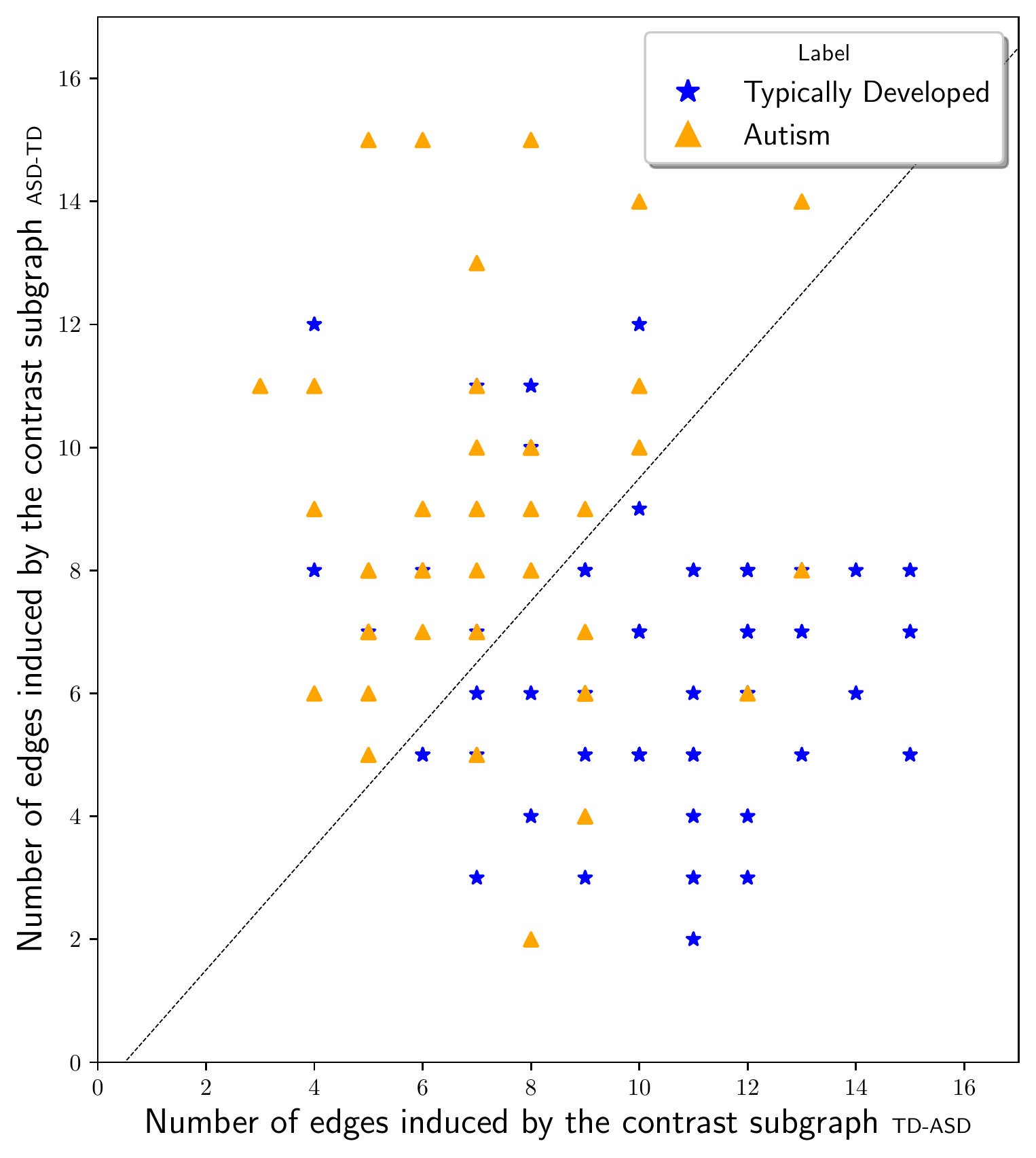}
	\end{minipage}
	\caption{Contrast subgraphs from the \highlight{Children} dataset
	(obtained with a larger value of $\alpha$ w.r.t \reffig{results_large}). Top left: \td-\asd contrast subgraph. Bottom left: \asd-\td contrast subgraph. Size of nodes and thickness of edges are respectively proportional to their degree and weight in the \emph{difference graph}. Edges shown are the ones whose weight is greater than 0.1. On the right, scatter plot showing for each individual the number of edges present in the subgraph induced by the contrast subgraph \td-\asd ($x$-axis) and by the contrast subgraph \asd-\td
	($y$-axis).
\label{fig:results_small}}
\end{figure}

Khan et al.~\cite{khan2015cerebellum} have discovered high connectivity
between cerebellum and prefrontal cortex,
posterior parietal cortex,
inferior and middle temporal gyri,
as a fingerprint distinguishing \td children from children with \asd.
These findings are consistent with the contrast subgraph \td-\asd in \reffig{results_large}.
Di Martino et al.~\cite{dimartino2011striatum} have discovered hyper-connectivity in
\asd children between the regions of striatum and the limbic cortex,
which is consistent with the contrast subgraph \asd-\td in \reffig{results_large}.

Furthermore, by inspecting \reffig{results_small}
we are able to discover another insight:
the ubiquitous presence of {\roi}s belonging to the left hemisphere
in the contrast subgraph {\td}-{\asd}.
This observation, again, is coherent with neuroscience findings:
the effect of some neural function being specialized in only one side of the brain
is called left/right lateralization of the brain.
For example, when it comes to the task of speech production,
it has been shown that most humans exhibit a left lateralization of the brain,
however, it has been reported that \asd patients show significantly
reduced left lateralization in speech-production-related
{\roi}s~\cite{nielsen2014leftimpair, eyler2012leftimpair, kleinhans2008leftimpair},
such as inferior frontal gyrus, inferior parietal lobule, and superior temporal gyrus,
which mostly appear also in the solution we provide.

\begin{figure}[t]
	\begin{minipage}[c]{.3\linewidth}
		\begin{tabular}{c}
			\hspace{-4mm}\includegraphics[width= 1.12\linewidth]{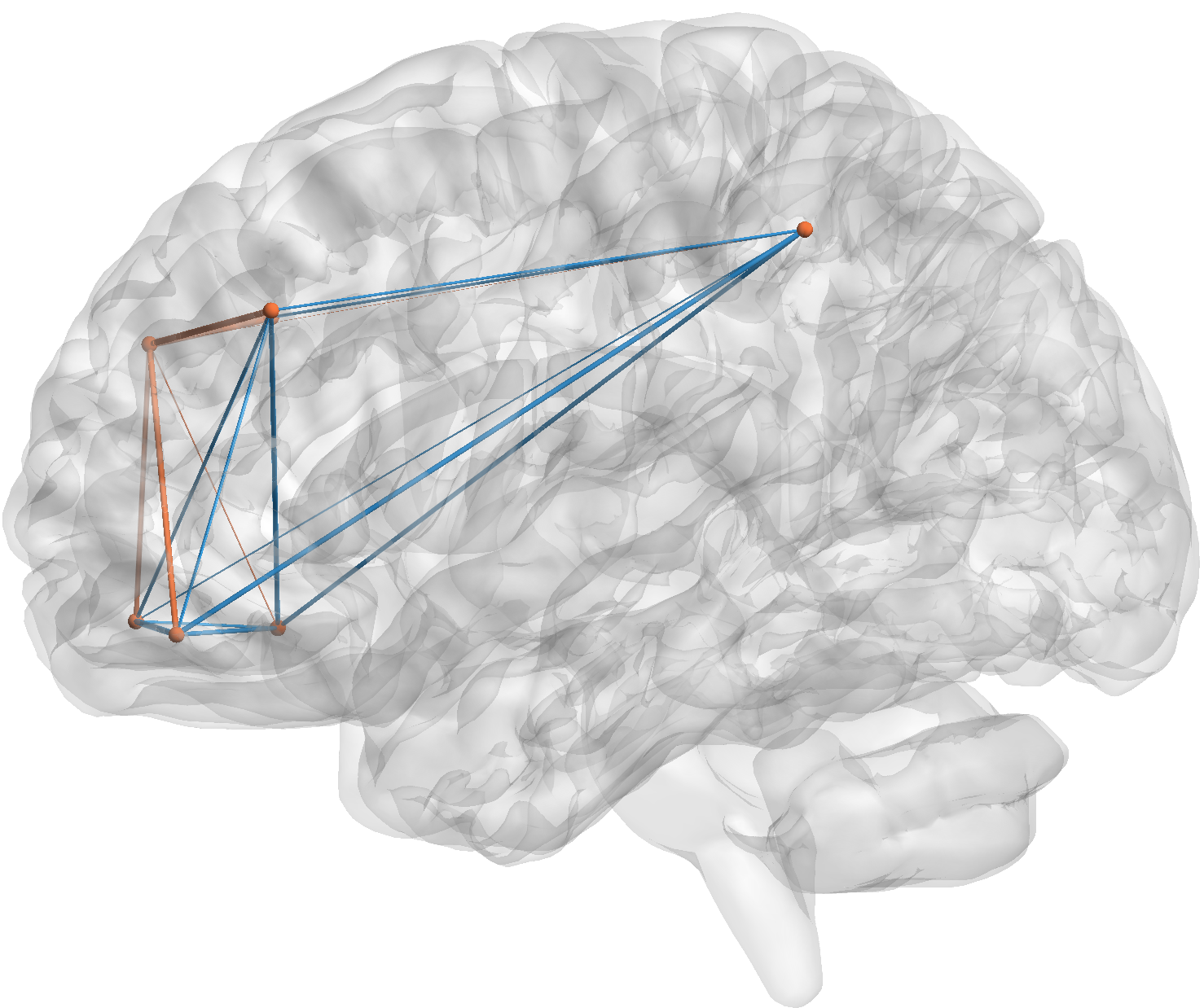}\\
			\hspace{-4mm}\includegraphics[width= \linewidth]{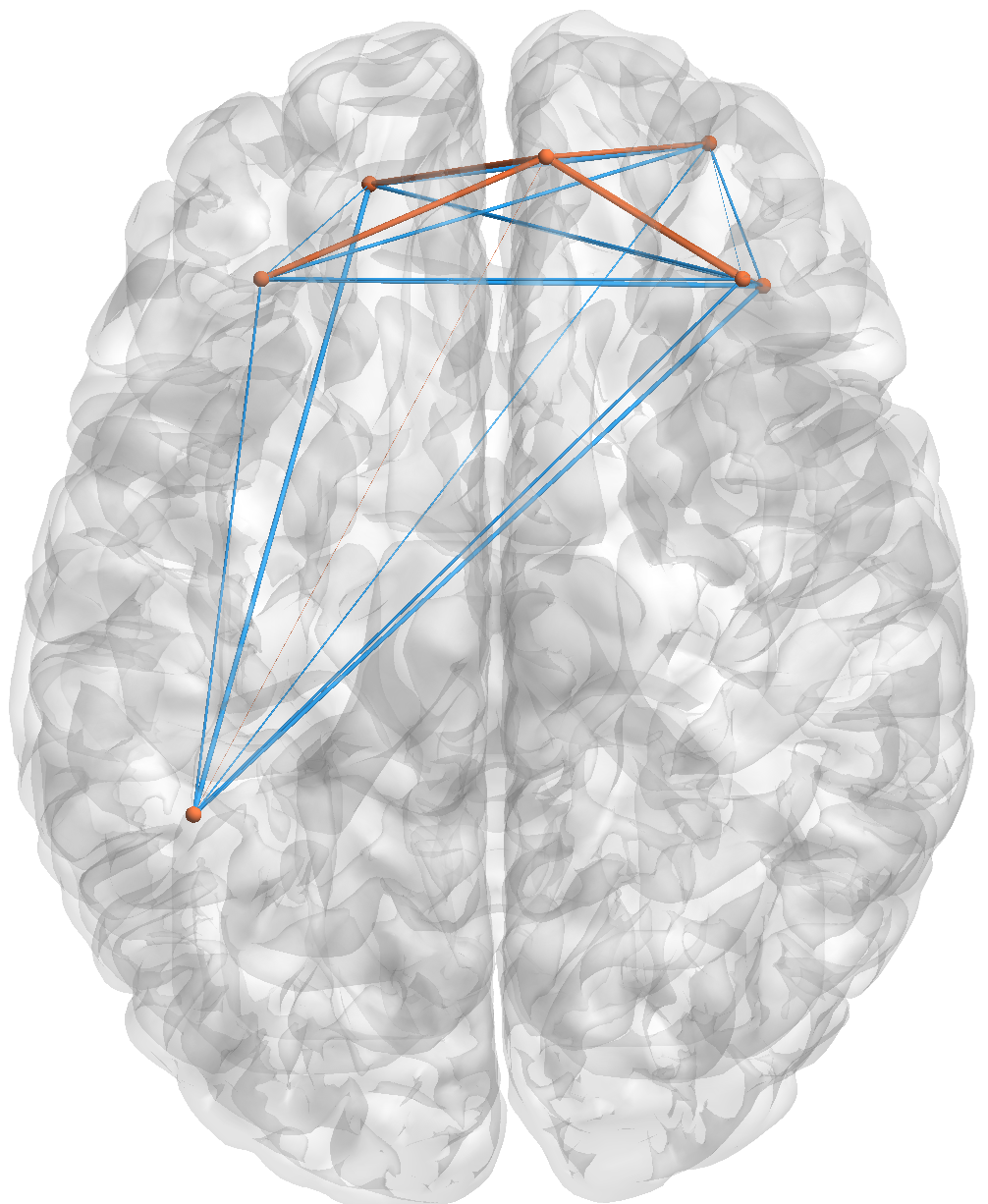}
		\end{tabular}
	\end{minipage}
	\begin{minipage}[c]{.68\linewidth}
		\includegraphics[width= \linewidth]{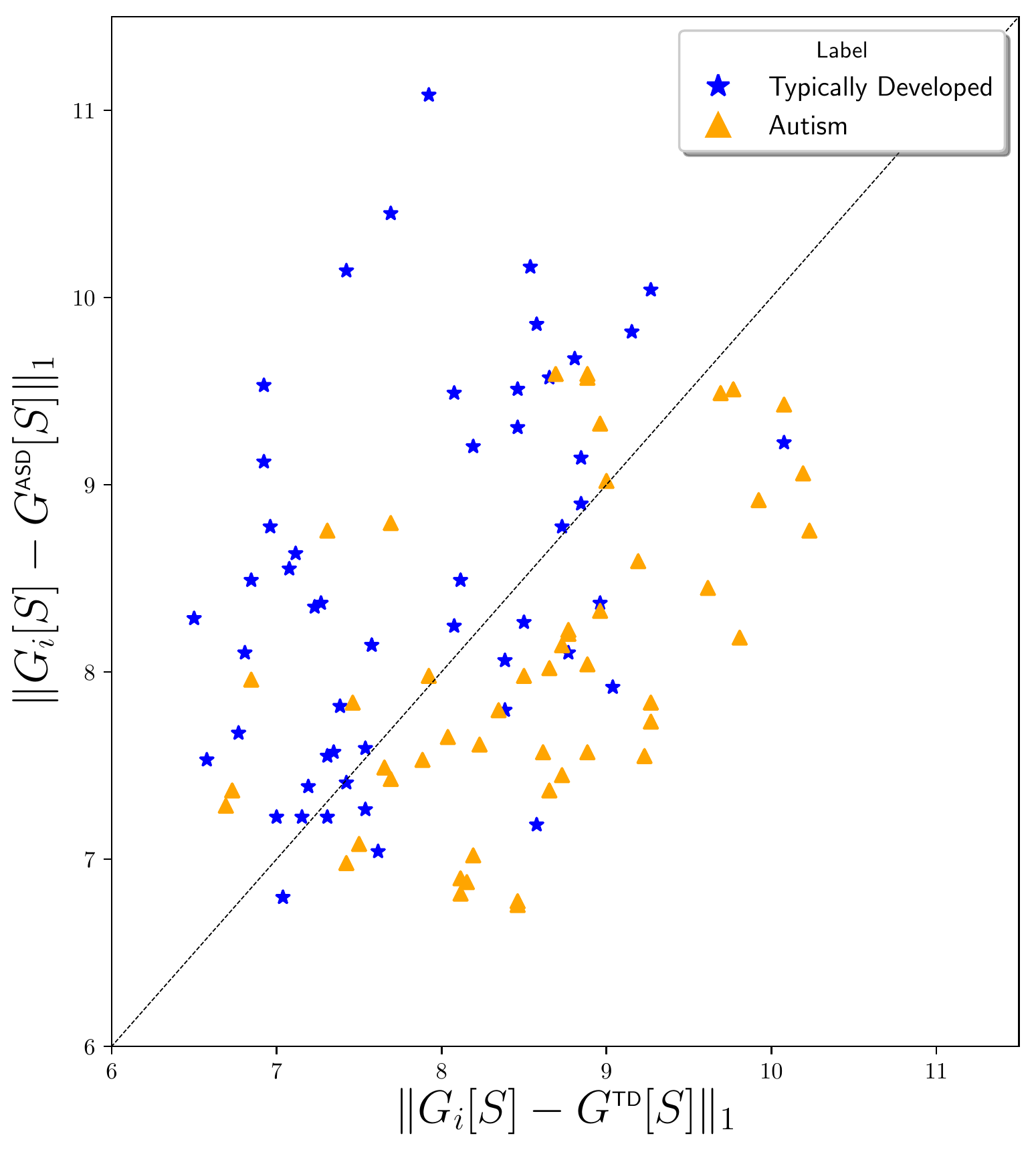}
	\end{minipage}
	\caption{Example of contrast subgraphs from the \highlight{Children} dataset according to \refpr{contrast_variant}. Top left: sagittal view. Bottom left: axial view. Thickness of edges is proportional to the absolute value of their weight in the \emph{difference graph}. Blue edges represent the ones that have positive weight in \diffnet{\td}{\asd}, orange ones in \diffnet{\asd}{\td}. Edges shown are the ones whose weight in absolute value is greater than 0.1.
	Such formulation of the problem highlights again that prefrontal cortex and posterior parietal cortex are regions of hyper-connectivity in \td patients, but also that at the same time such regions shows hyper-connectivity with medial prefrontal cortex in \asd patients, and hypo-connectivity in \td ones. Such contemporaneity is caught by the maximization of specularity that \refpr{contrast_variant} operates, instead of taking into account difference.
		\label{fig:results_alt}}
\end{figure}

\spara{Symmetric contrast subgraph (\refpr{contrast_variant}).}
In the case of the symmetric problem definition,
we extract a single contrast subgraph for the two classes.
An example is shown in \reffig{results_alt}:
blue edges show edges with positive weight in \diffnet{\td}{\asd},
and orange edges show edges with positive weight in \diffnet{\asd}{\td}.
There are seven {\roi}s involved in this solution:
five of those already appeared in the contrast subgraphs {\td}-{\asd},
previously shown in \reffig{results_large} and \reffig{results_small}.
All these five {\roi}s belong either to prefrontal cortex or to posterior parietal cortex,
which confirms that such {\roi}s are discriminative for the two classes.
The solution is completed by right middle frontal gyrus and superior frontal gyrus medial.
This finding is coherent with the results of Gilbert et al.~\cite{gilbert2008prefrontal}
who provide evidence that the medial prefrontal cortex has an anomalous behavior
in \asd patients during several tasks.

Given that in this formulation of the problem
there is a single contrast subgraph $S$ for the two classes,
to produce features to classify a patient $G_i$
we compute the $L_1$ norm between the subgraph induced by $S$ in $G_i$
and the subgraph induced by $S$ in \GTD and \GASD, respectively.
The results are shown in the scatter plot of \reffig{results_alt} (right),
where we show each graph $G_i$ with value $\|G_i[S] - \GTD[S]\|_1$ on the $x$-axis and
value $\|G_i[S] - \GASD[S]\|_1$ on the $y$-axis
Here, $G[S]$ denotes the subgraph induced by $S$ in $G$.
Also in this case, we can see a very good separability defined by the diagonal.

%% file: experiments2.tex
\subsection{Classification}
\label{sec:experiments}

We evaluate the effectiveness of contrast subgraphs in classifying
brain networks according to the target classes \td and \asd.
We compare with four baselines:
\highlight{graph2vec} \cite{narayanan2017graph2vec},
\highlight{sub2vec} \cite{adhikari2018sub2vec},
\highlight{WL-Kern} \cite{shervashidze2011weisfeiler}, and
\highlight{Embs} \cite{gutierrez2019embedding}.
All baselines provide a feature map, which we use for the classification task.
Details about the implementation and parameter tuning of these methods are reported in
Appendix~\ref{appendix:params}.

For our method,
we use the same features as 
in \reffig{results_small} (we refer to this approach as \highlight{CS-P1}) and
\reffig{results_alt} (we refer to this approach as \highlight{CS-P2}):
under both settings our method uses only two~features.

We employ the SVM-based classification process of Narayanan et al.~\cite{narayanan2017graph2vec}.
We randomly split the data into 80/20 training/test subsets.
Using the training set and 5-fold cross validation we select the best hyperparameters for the classifier.
We then apply the best classifier to the test set.
We repeat this process 5 times for each method.

We report the average accuracy for all the classification experiments in Table~\ref{tab:res1}.
We can observe the overall good performance of our method in all these tasks.
Moreover, a number of observations can be drawn:

\begin{description}
\item[Only two features:]
used as a proof of concept,
our method was tested employing only two features derived by contrast subgraphs.
In more complex classification tasks, however,
one can consider additional features from different contrast subgraphs
so as to improve the classification accuracy.

\smallskip

\item[A single parameter:]
our method has only one parameter to tune, i.e., $\alpha$,
while the other methods (with the exception of \highlight{WL-Kern}) require tuning of many parameters.

\smallskip

\item[Run-time efficiency:] extracting the contrast subgraphs from this type of brain connectome datasets
always take less than 30 seconds,
while the other methods (with the exception of \highlight{WL-Kern}) have longer running times.

\smallskip

\item[Explainability:]
in the pipeline we presented,
our method uses only two simple features,
instead of embedded features or convoluted kernels.
These features enable simple description in form of rule, as well as visualization, of the
 decision boundary discriminating between the two classes.
\end{description}
\smallskip

Figures~\ref{fig:appendix} and \ref{fig:appendix_alt} reports the decision boundaries for \highlight{CS-P1}  and \highlight{CS-P2} (respectively).

\begin{table}[t!]
	\caption{\label{tab:res1}Results of the experiments performed over the datasets described in \reftab{datasets}. Each value represents the average accuracy, along with its relative standard deviation.}

\vspace{-2mm}

	\begin{tabular}{rllll}
	\toprule
\multicolumn{1}{c}{}		& \multicolumn{1}{c}{\highlight{Children}} & \multicolumn{1}{c}{\highlight{Adolescents}} & \multicolumn{1}{c}{\highlight{EyesClosed}} & \multicolumn{1}{c}{\highlight{Male}}\\
	\midrule
		\highlight{CS-P1} & $\textbf{0.86} \pm 0.07$ & $\textbf{0.72} \pm 0.07$  & $0.71 \pm 0.03$ & $0.63 \pm 0.01$ \\ 
		\highlight{CS-P2} & $\textbf{0.86} \pm 0.04$ & $0.71 \pm 0.04$ & $\textbf{0.72} \pm 0.08$ & $\textbf{0.65} \pm 0.03$ \\ 
		\highlight{graph2vec} & $0.72 \pm 0.13$ & $0.65 \pm 0.05$ & $0.60 \pm 0.04$ & $0.56 \pm 0.02$ \\ 
		\highlight{sub2vec} & $0.66 \pm 0.01$ & $0.59 \pm 0.04$ & $0.60 \pm 0.01$  & $0.57 \pm 0.01$ \\ 
		\highlight{WL-Kern} & $0.52 \pm 0$ & $0.52 \pm 0$& $0.54 \pm 0$ & $0.50 \pm 0$ \\ 
		\highlight{Embs} & $0.70 \pm 0.12$ & $0.58 \pm 0.03$ & $0.59 \pm 0.04$  & $0.57 \pm 0.03$ \\ 	
	\bottomrule
	\end{tabular}
\end{table}

\begin{figure}[t!]
		\vspace{2mm}
	\centering
	\hspace{-4mm}\includegraphics[width=1.05\linewidth]{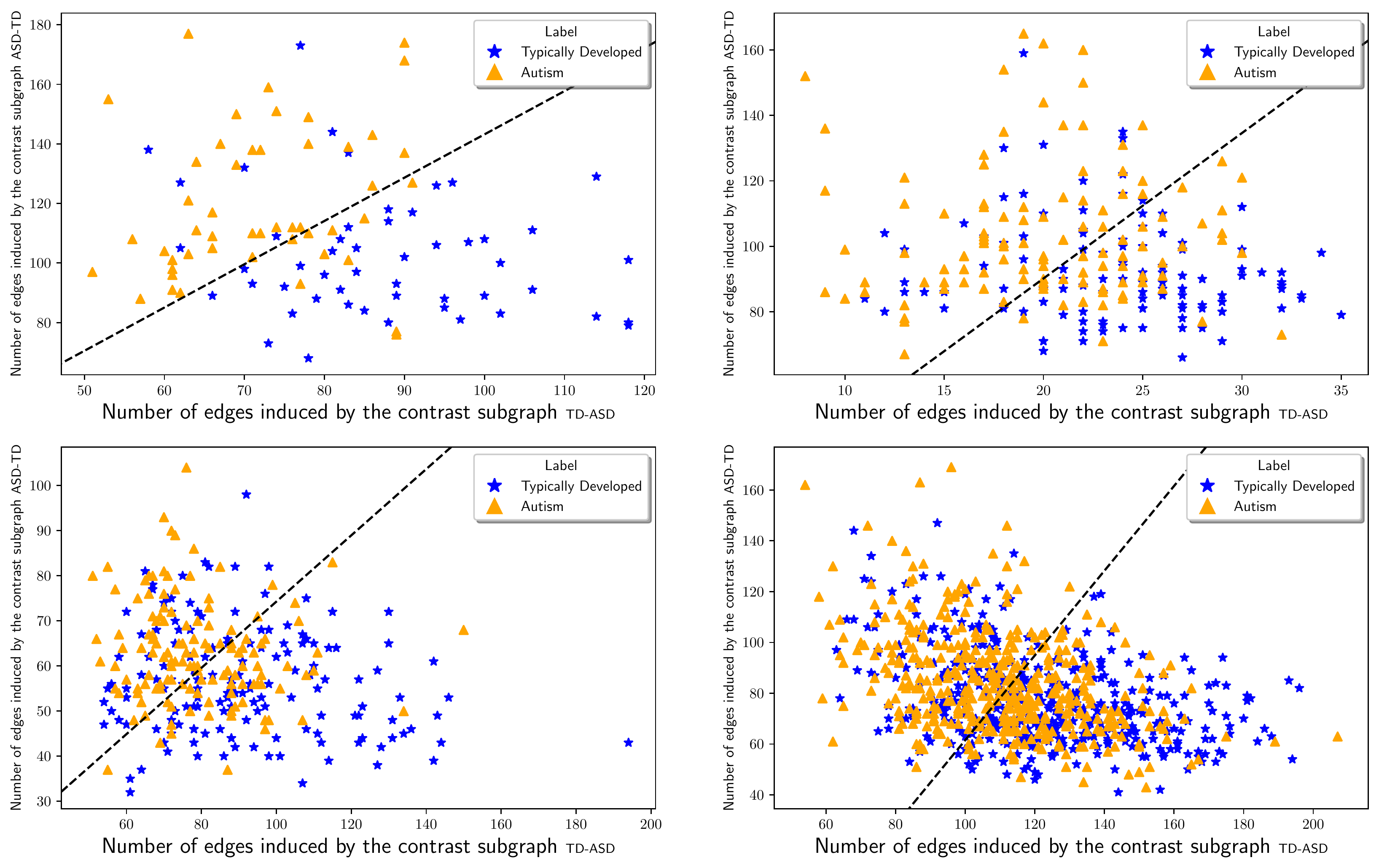}
	\vspace{-4mm}
	\caption{Decision boundaries for \highlight{CS-P1}. Top \highlight{Children} (left), \highlight{Adolescents} (right). Bottom \highlight{EyesClosed} (left), \highlight{Male} (right). 	\label{fig:appendix}}
		\vspace{2mm}
\end{figure}

\begin{figure}[t!]
	\centering
	\hspace{-4mm}\includegraphics[width=1.05\linewidth]{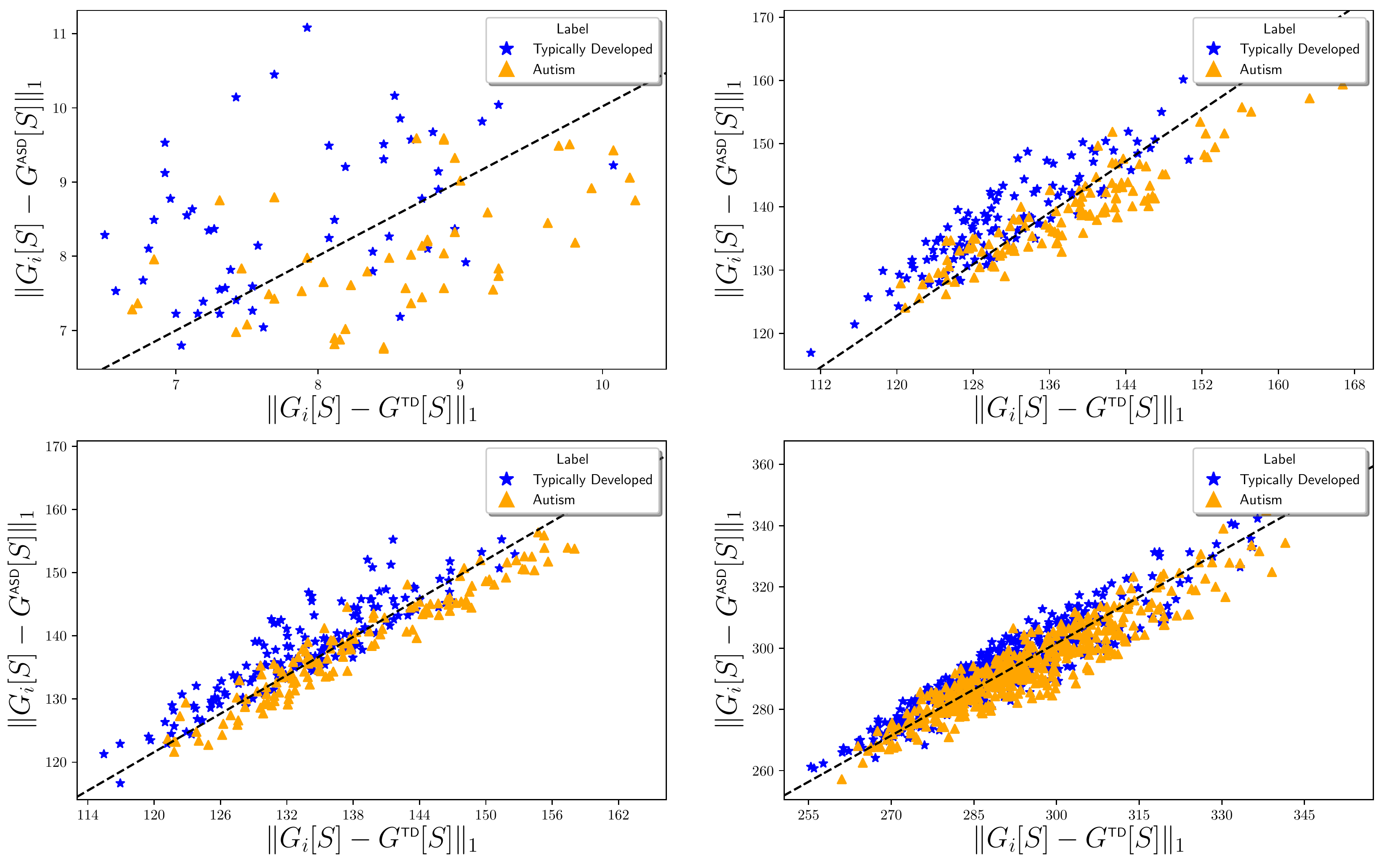}
	\vspace{-4mm}
	\caption{Decision boundaries for \highlight{CS-P2}. Top \highlight{Children} (left), \highlight{Adolescents} (right). Bottom \highlight{EyesClosed} (left), \highlight{Male} (right).  	\label{fig:appendix_alt}}
\end{figure}

%% file: conclusions.tex
Learning models that are able to discriminate brains of patients affected by a mental disorder from those of healthy individuals, is attracting a lot of interest. However, often the accuracy of the models is a predominant goal over its interpretability, which is instead a key requirement in neuroscience.

In this paper we approach the task of brain-network classification with a two-fold goal:
to achieve good accuracy, but most importantly,
to identify discriminant brain patterns that lead a model to classify an individual.
The framework we propose,
based on the notion of \emph{contrast subgraph}, satisfies both of these conditions.
It outperforms several state-of-the-art competitors and returns
some very intuitive explanations, which can be shared with experts from the neuroscience field. Moreover, contrast subgraphs are exceptionally easy to compute, both in terms of runtime and memory.

\spara{Future work.} In this first work we focus on Autism Spectrum Disorder mostly because of the wide availability of public datasets. However, the proposed approach can be applied in several other contexts and different types of mental disorders. To push forward this research direction,
we have created an ongoing collaboration
with neuroscientists.
We are currently working with the domain experts to create
datasets with patients affected by Bipolar Disorder and by Schizophrenia.

On the algorithmic side, we plan to investigate the extraction of the top-$k$ contrast subgraphs, following an approach similar to that of Balalau et al.~\cite{BalalauBCGS15} for densest subgraphs.
Exploiting the information given by multiple contrast subgraphs, could help both to improve the classification performances, and to detect further interesting patterns that may not appear in the single solution. 

%% file: acknowledgements.tex
The authors wish to thank Edoardo Galimberti for many meaningful discussions in the early stage of this project.
Francesco Bonchi acknowledges support from Intesa Sanpaolo Innovation Center.
Aristides Gionis is supported by three Academy of Finland projects (286211, 313927, 317085),
the ERC Advanced Grant REBOUND (834862),
the EC H2020 RIA project ``SoBigData++'' (871042), and the
Wallenberg AI, Autonomous Systems and Software Program (WASP).
The funders had no role in study design, data collection and analysis, decision to publish, or preparation of the manuscript. 

%% file: appendix.tex
\subsection{Pseudocode}
\label{appendix:pseudocode}

We provide pseudocodes of the algorithms used for computing contrast subgraphs,
according to Section \ref{sec:algorithms}.
In particular,
Algorithm~\ref{alg:cs} solves \refpr{contrast},
Algorithm~\ref{alg:cs_alt} solves \refpr{contrast_variant},
Algorithm~\ref{alg:densdp} displays the algorithm DENSDP by Cadena et al.~\cite{cadena2016dense} based on the semi-definite programming (SDP) routine,
while Algorithm~\ref{alg:localsearch} describes the local-search routine by Tsourakakis et al.~\cite{tsourakakis2013denser}, which refines the solution produced by the SDP method.

\smallskip

\begin{mdframed}[innerbottommargin=3pt,innertopmargin=3pt,innerleftmargin=6pt,innerrightmargin=6pt,backgroundcolor=gray!10,roundcorner=10pt]
All our code and datasets are made available at:
{\center {\url{https://github.com/tlancian/contrast-subgraph}}}
\end{mdframed}

Algorithms \ref{alg:cs} and \ref{alg:cs_alt} simply apply Proposition \ref{fact:mapping} in order to map
Problem \ref{prb:contrast} and \ref{prb:contrast_variant} (respectively) to Problem \ref{prb:edsn}, and then invoke Algorithm \ref{alg:densdp} over the difference graph $\diffnet{\calA}{\calB}$ and with the $\overline{\alpha}$ function taking constantly the $\alpha$ parameter of our problems.

The subroutine dubbed SDP in Algorithm \ref{alg:densdp}, corresponds to the following semi-definite programming:

\begin{equation*}
\begin{aligned}
& \underset{(u,z) \in E}{\text{max}}
& & \sum_{(u,z) \in E} w(u,z)\Big(\dfrac{1+\vec{v}_u \vec{v}_0 + \vec{v}_z \vec{v}_0 + \vec{v}_u \vec{v}_z}{4}\Big) -\\
& & & \sum_{u,z \in V, u \neq z} \overline{\alpha}(u,z)\Big(\dfrac{1+\vec{v}_u \vec{v}_0 + \vec{v}_z \vec{v}_0 + \vec{v}_u \vec{v}_z}{4}\Big) \\
& \text{s.t.} &  &  & \\
& &  & \vec{v}_{u}^T \cdot \vec{v}_u = 1 & \\
& &  & \vec{v}_0, \vec{v}_u \in \mathbb{R}^{n+1}. & \\
\end{aligned}
\end{equation*}

The function $f_{\alpha}(\cdot)$ in Algorithm \ref{alg:localsearch} corresponds to the notion of \emph{edge-surplus},
mentioned in \refsec{algorithms}, Equation~(\ref{eq:surplus}),
by Tsourakakis et al.~\cite{tsourakakis2013denser}.

\begin{algorithm}[h]
	\caption{Contrast Subgraph (\refpr{contrast})} \label{alg:cs}
	\begin{algorithmic}[1]
		\REQUIRE $\calA=\{G_1^\A,\ldots,G_{r_\A}^\A\}$; $\calB=\{G_1^\B,\ldots,G_{r_\B}^\B\}$; $\alpha \in \realsnn$
		\ENSURE $S^* \subseteq V$, i.e., the solution to \refpr{contrast}
\FORALL{$(u,v) \in V \times V$}		
\STATE $\wA(u,v) \gets \frac{1}{r_{\A}} \left| G_i^\A \in \calA \text{ s.t. }  (u,v)\in E_i^\A \right|$ $\;\;$ \hfill\COMMENT{\texttt{Eq.(\ref{eq:summary})}}
\STATE $\wB(u,v) \gets \frac{1}{r_{\B}} \left| G_i^\B \in \calB \text{ s.t. }  (u,v)\in E_i^\B \right|$
\ENDFOR
		\STATE \diffnet{\calA}{\calB} $\gets (V, \wA-\wB)$ \hfill\COMMENT{\texttt{Prop. \ref{fact:mapping}}}
		\STATE $S^* \gets$ DENSDP(\diffnet{\calA}{\calB},  $\overline{\alpha}(\cdot) = \alpha$)  \hfill\COMMENT{\texttt{Alg.\ref{alg:densdp}}}
	\end{algorithmic}
\end{algorithm}

\begin{algorithm}[h]
	\caption{Symmetric Contrast Subgraph (\refpr{contrast_variant})} \label{alg:cs_alt}
	\begin{algorithmic}[1]
				\REQUIRE $\calA=\{G_1^\A,\ldots,G_{r_\A}^\A\}$; $\calB=\{G_1^\B,\ldots,G_{r_\B}^\B\}$; $\alpha \in \realsnn$
		\ENSURE $S^* \subseteq V$, i.e., the solution to \refpr{contrast_variant}
		\FORALL{$(u,v) \in V \times V$}		
\STATE $\wA(u,v) \gets \frac{1}{r_{\A}} \left| G_i^\A \in \calA \text{ s.t. }  (u,v)\in E_i^\A \right|$  \hfill\COMMENT{\texttt{Eq.(\ref{eq:summary})}}
\STATE $\wB(u,v) \gets \frac{1}{r_{\B}} \left| G_i^\B \in \calB \text{ s.t. }  (u,v)\in E_i^\B \right|$
\ENDFOR
		\STATE \diffnet{\calA}{\calB} $\gets (V, |\wA-\wB|)$  \hfill\COMMENT{\texttt{Prop. \ref{fact:mapping}}}
		\STATE $S^* \gets$ DENSDP(\diffnet{\calA}{\calB}, $\overline{\alpha}(\cdot) = \alpha$) \hfill\COMMENT{\texttt{Alg.\ref{alg:densdp}}}
	\end{algorithmic}
\end{algorithm}

\begin{algorithm}[h]
	\caption{DENSDP \cite{cadena2016dense}} \label{alg:densdp}
	\begin{algorithmic}[1]
		\REQUIRE Weighted graph $G=(V,w)$; function $\overline{\alpha}: V \times V \rightarrow \realsnn$
		\ENSURE $S^* \subseteq V$, i.e., the solution to Problem~\ref{prb:edsn} 
		
		\STATE $ \mathcal{V} \gets \text{SDP}(G, \overline{\alpha}), \text{ where }
		  \mathcal{V} = \{\vec{v}_0\} \cup \{ \vec{v}_u \mid u \in V\}$
		\STATE Sample $\vec{r}$ from $\mathcal{N} \sim (\vec{0}_{(|V|+1)}, I_{(|V|+1)\times(|V|+1)})$
		\STATE $T \gets \sqrt{4 \log |V|}$
		\FOR {$u \in V$}
			\STATE $z_u \gets (\vec{v}_u \cdot \vec{r}) /T$
			\STATE \textbf{if} $|z_u|>1$ \textbf{then} $y_u \gets z_u/|z_u|$ \textbf{else} $y_u \gets z_u$			
			\STATE $x_u \gets \begin{cases}
1, & \text{ with probability } \frac{1+y_u}{2}\\
-1 & \text{ with probability } \frac{1-y_u}{2}
\end{cases}
$
		\ENDFOR
		
		\STATE $S' \gets \{u \mid x_u = 1\}$
		\STATE $S \gets $LocalSearch$(G' = (S',w), \alpha = \overline{\alpha}(\cdot))$ \hfill\COMMENT{\texttt{Alg.(\ref{alg:localsearch})}}
	
	\end{algorithmic}
\end{algorithm}

\begin{algorithm}[h]
	\caption{LocalSearch \cite{tsourakakis2013denser}} \label{alg:localsearch}
	\begin{algorithmic}[1]
		\REQUIRE Weighted graph $G=(V,w)$, $\alpha \in \realsnn$,  maximum number of iterations $T_{MAX}$
		\ENSURE $S^* \subseteq V$
		
		\STATE $S \gets \{v\}$, where $v$ is chosen uniformly at random
		\STATE $b_1 \gets$ TRUE, $t \gets 1$
		\WHILE {$b_1$ \textbf{and} $t \leq T_{MAX}$}
		
			\STATE $b_2 \gets $TRUE
			\WHILE {$b_2$}
				\IF {$\exists u \in V \setminus S$ s.t. $f_{\alpha}(S \cup \{u\}) \geq f_{\alpha}(S)$}
			
				\STATE $S \gets S \cup \{u\}$
			
				\ELSE
			
				\STATE $b_2 \gets$ FALSE
			
				\ENDIF

			\ENDWHILE
		
			\IF {$\exists u \in S$ s.t. $f_{\alpha}(S \setminus \{u\}) \geq f_{\alpha}(S)$}
		
			\STATE $S \gets S \setminus \{u\}$
		
			\ELSE
			\STATE $b_1 \gets$ FALSE
			\STATE $t \gets t+1$
			\ENDIF
		
		\ENDWHILE
		
		\STATE $S^* \gets \argmax_{\hat{S} \in \{S,V \setminus S\}} f_{\alpha}(\hat{S})$
	\end{algorithmic}
\end{algorithm}

\subsection{Competitors and parameters tuning}\label{appendix:params}

In this section we provide further details of the experimental setup we described in \refsec{experiments}.

\spara{Embeddings.}
In order to obtain embeddings for the baseline methods,
we followed different strategies.
\highlight{graph2vec} \cite{narayanan2017graph2vec} and
\highlight{Embs} \cite{gutierrez2019embedding}
provide directly embeddings for a given set of networks.
\highlight{sub2vec} \cite{adhikari2018sub2vec} provides an embedding for a set of subgraphs in a single network, therefore we employed the strategy of Adhikari et al~\cite{adhikari2018sub2vec},
and consider any brain network to be a specific subgraph of a single network made
by the union of all these subgraphs.
Regarding \highlight{WL-Kern} \cite{shervashidze2011weisfeiler},
we employed the feature map it produces for computing subsequently the kernel matrix.

\spara{Code source and parameters.} 
We next provide links to source code of all the algorithms compared, together with the parameter settings used in the experiments. The ones not mentioned in this section,
have been left to their default value in their respective sources. We have approached to the parameter tuning, taking values around the default value set by the respective authors. For every algorithms listed here, we have performed the experiment for any possible combination of this parameters, retaining only the one that led to the best results, in term of average accuracy over 5 experiments.

\spara{\highlight{CS-P1}}

\noindent Source: \url{https://github.com/tlancian/contrast-subgraph}

\noindent Parameters tuned:

\begin{itemize}
	\item $\alpha \text{ for \td-\asd}: \{70, 75, 80, 85, 90, 93, 95, 97\}$
	\item $\alpha \text{ for \asd-\td}: \{70, 75, 80, 85, 90, 93, 95, 97\}$
\end{itemize}

\noindent Best Parameters:

\begin{table}[h]
	
	\vspace{-2mm}
	
	\begin{tabular}{rcccc}
		\toprule
		\multicolumn{1}{c}{}		& \multicolumn{1}{c}{\highlight{Children}} & \multicolumn{1}{c}{\highlight{Adolescents}} & \multicolumn{1}{c}{\highlight{EyesClosed}} & \multicolumn{1}{c}{\highlight{Male}}\\
		\midrule
		 $\alpha \text{ for \td-\asd}$ & 70 & 95 & 75 & 75  \\ 
		 $\alpha \text{ for \asd-\td}$ & 80 & 70 & 75 & 70 \\ 
		\bottomrule
	\end{tabular}
\end{table}

In addition, we report that the $\alpha$ values used to produce Figure~\ref{fig:results_large} are the same that resulted to be the best results for \highlight{Children}. The ones for \reffig{results_small} are respectively 90 and 93.

\spara{\highlight{CS-P2}}

\noindent Source: \url{https://github.com/tlancian/contrast-subgraph}

\noindent Parameters tuned:

\begin{itemize}
	\item $\alpha: \{70, 75, 80, 85, 90, 93, 95, 97\}$
\end{itemize}

\noindent Best Parameters:

\begin{table}[h]
	
	\vspace{-2mm}
	
	\begin{tabular}{rcccc}
		\toprule
		\multicolumn{1}{c}{}		& \multicolumn{1}{c}{\highlight{Children}} & \multicolumn{1}{c}{\highlight{Adolescents}} & \multicolumn{1}{c}{\highlight{EyesClosed}} & \multicolumn{1}{c}{\highlight{Male}}\\
		\midrule
		$\alpha$ & 70 & 75 & 70 & 70  \\ 
		\bottomrule
	\end{tabular}
\end{table}

In addition, we report that the $\alpha$ used to produce Figure~\ref{fig:results_alt} is the same that resulted to be the best results for \highlight{Children}.

\spara{\highlight{graph2vec} \cite{narayanan2017graph2vec}}

\noindent Source: \url{https://github.com/MLDroid/graph2vec_tf}

\noindent Parameters tuned:

\begin{itemize}
	\item Batch Size: \{64, 128, 256, 512\}
	\item Epochs: \{100\}
	\item Embedding Size: \{256, 512, 1024, 2048\}
	\item Negative Samples: \{10, 20, 30, 40\}
	\item Height WL Kernel: \{3, 4, 5\}
\end{itemize}

\noindent Best Parameters:

\begin{table}[h]

	\vspace{-2mm}
	
	\begin{tabular}{rcccc}
		\toprule
		\multicolumn{1}{c}{}		& \multicolumn{1}{c}{\highlight{Children}} & \multicolumn{1}{c}{\highlight{Adolescents}} & \multicolumn{1}{c}{\highlight{EyesClosed}} & \multicolumn{1}{c}{\highlight{Male}}\\
		\midrule
		Batch Size & 256 & 256 & 64 & 128  \\ 
		Embedding Size & 256 & 2048 & 512 & 1024 \\ 
		Negative Samples & 30 & 30 & 30 & 20 \\ 
		Height WL Kernel & 3 & 3 & 3 & 3 \\ 
		\bottomrule
	\end{tabular}
\end{table}

\pagebreak

\spara{\highlight{sub2vec} \cite{adhikari2018sub2vec}}

\noindent Source: \url{https://goo.gl/Ef4q8g}

\noindent Parameters tuned:

\begin{itemize}
	\item Property: \{Neighborhood, Structural\}
	\item Walk Length: \{100000, 150000, 200000\}
	\item Embedding Size: \{128, 256, 512\}
	\item Model: \{DBON, DM\}
	\item Training Iterations: \{10, 20\}
\end{itemize}

\noindent Best Parameters:

\begin{table}[h]
	
	\vspace{-2mm}
	
	\begin{tabular}{rcccc}
		\toprule
		\multicolumn{1}{c}{}		& \multicolumn{1}{c}{\highlight{Children}} & \multicolumn{1}{c}{\highlight{Adolescents}} & \multicolumn{1}{c}{\highlight{EyesClosed}} & \multicolumn{1}{c}{\highlight{Male}}\\
		\midrule
		Property & S & S & S & S  \\ 
		Walk Length & 100000 & 150000 & 150000 & 150000 \\ 
		Embedding Size & 512 & 128 & 128 & 128 \\ 
		Model & DM & DM & DBON & DBON \\ 
		Training Iterations & 10 & 20 & 20 & 20 \\ 
		\bottomrule
	\end{tabular}
\end{table}

\spara{\highlight{Embs} \cite{gutierrez2019embedding}}

\noindent Source: \url{https://github.com/leoguti85/GraphEmbs}

\noindent Parameters tuned:

\begin{itemize}
	\item Epochs: \{100, 300, 500\}
	\item Batch Size: \{64, 128, 256\}
	\item Embedding Size: \{128, 256, 512, 1024\}
	\item Noise: \{0.3, 0.5\}
\end{itemize}

\noindent Best Parameters:

\begin{table}[h]
	
	\vspace{-2mm}
	
	\begin{tabular}{rcccc}
		\toprule
		\multicolumn{1}{c}{}		& \multicolumn{1}{c}{\highlight{Children}} & \multicolumn{1}{c}{\highlight{Adolescents}} & \multicolumn{1}{c}{\highlight{EyesClosed}} & \multicolumn{1}{c}{\highlight{Male}}\\
		\midrule
		Epochs & 500 & 500 & 300 & 100  \\ 
		Batch Size & 64 & 128 & 256 & 256 \\ 
		Embedding Size & 1024 & 1024 & 1024 & 1024 \\ 
		Noise & 0.5 & 0.3 & 0.5 & 0.3 \\ 
		\bottomrule
	\end{tabular}
\end{table}

\spara{\highlight{WL-Kern} \cite{shervashidze2011weisfeiler}}

\noindent Source: \url{http://mlcb.is.tuebingen.mpg.de/Mitarbeiter/Nino/WL}

\noindent Parameters tuned:

\begin{itemize}
	\item $h: \{1, 2, 3, 4, 5\}$
\end{itemize}

\noindent Best Parameters:

\begin{table}[h]
	
	\vspace{-2mm}
	
	\begin{tabular}{rcccc}
		\toprule
		\multicolumn{1}{c}{}		& \multicolumn{1}{c}{\highlight{Children}} & \multicolumn{1}{c}{\highlight{Adolescents}} & \multicolumn{1}{c}{\highlight{EyesClosed}} & \multicolumn{1}{c}{\highlight{Male}}\\
		\midrule
		h & 3 & 3 & 3 & 3  \\ 
		\bottomrule
	\end{tabular}
\end{table}